\newif\ifSC
\SCtrue
\ifSC
\documentclass[onecolumn,draftclsnofoot,12pt]{IEEEtran}
\else
\documentclass[twocolumn,10pt]{IEEEtran}
\fi
\usepackage{subfigure}
\usepackage{etoolbox}

\newtoggle{SC}
\togglefalse{SC}
\ifSC
\toggletrue{SC}
\fi

\usepackage{bbm}
\usepackage{amssymb}
\usepackage{amsmath,amsthm}
\usepackage{color}
\usepackage{amsfonts}
\usepackage{graphicx}
\usepackage{verbatim}
\usepackage{epstopdf}
\usepackage{cite}
\usepackage{tabulary}
\usepackage{mathtools}

\def\primary{primary }
\def\secondary{secondary }
\newcommand{\Primary}{Primary }
\newcommand{\Secondary}{Secondary }

\newcommand{\fracS}[2]{#1/#2}
\newcommand\expect[1]{\mathbb{E}\left[#1\right]}
\newcommand\prob[1]{\mathbb{P}\left[#1\right]}
\newcommand\ind[1]{\mathbbm{1}_{#1}}
\newcommand\indside[1]{\mathbbm{1}\left({#1}\right)}

\newcommand{\SINR}{\text{SINR}}

\newcommand{\expects}[2]{\mathbb{E}_{#1}\left[#2\right] }
\newcommand{\laplace}[1]{\mathcal{L}_{#1} }

\newcommand{\bcomment}[1]{\textcolor{blue}{[\underline{\bf Abhishek}: ] #1}}

\newcommand{\Pc}{\mathrm{P}}
\newcommand{\Rc}{\mathrm{R}}

\newcommand{\SThres}{\tau}
\newcommand{\RThres}{\rho}

 \newcommand{\intd}{\mathrm{d}}

\newcommand{\IT}{\xi}

\newcommand{\tx}{\mathrm{T}}
\newcommand{\rx}{\mathrm{R}}

\renewcommand{\L}{\mathrm{L}}
\newcommand{\N}{\mathrm{N}}

\newcommand{\pu}{\mathrm{P}}
\newcommand{\su}{\mathrm{S}}
\newcommand{\cu}{\mathrm{C}}

\newcommand\lambdaPT{\lambda_\pu}
\newcommand\lambdaPR{\mu_\pu}

\newcommand\lambdaST{\lambda_\su}
\newcommand\lambdaSR{\mu_\su}

\newcommand\lambdaR{\mu}

\newcommand\PPPT{\Phi_\pu}
\newcommand\PPPR{\Psi_\pu}

\newcommand\PPST{\Phi_\su}
\newcommand\PPSR{\Psi_\su}

\newcommand\PPR{\Psi}

\newcommand{\NoiseP}{\sigma^2_\pu}
\newcommand{\NoiseS}{\sigma^2_\su}
\newcommand{\powerP}[1]{P_{\pu#1}}
\newcommand{\powerS}[1]{P_{\su#1}}

\newcommand{\exclusionfun}[2]{E_{#1#2}} 

\newcommand{\UEP}{$\mathrm{UE}_\pu$}
\newcommand{\UES}{$\mathrm{UE}_\su$}

\newcommand{\XX}{\overline{P_{\su}}}
\newcommand{\XXi}[1]{\overline{P_{\su#1}}}

\newcommand{\beam}{\mathrm{b}}

\newcommand{\AntGain}[2]{G_{#1#2}}
\newcommand{\nAntP}{N_\mathrm{P}}
\newcommand{\nAntS}{N_\mathrm{S}}

\newcommand{\bigconditioned}{\left.\vphantom{\frac34}\right|}
\newcommand{\y}{{\mathbf{y}}}
\newcommand{\x}{{\mathbf{x}}}
\newcommand{\expU}[1]{e^{#1}}

\newcommand{\imdfun}[2]{Ko_{#1}\left(#2,\IT\right)}
\newcommand{\imdfuni}[2]{K_{#1}\left(#2\right)}

\newcommand{\imdfunP}[2]{M_{ #1}\left(#2\right)}

\newcommand{\ths}{\text{th}}

\def\home{\hbox{\kern3pt \vbox to13pt{}%
   \pdfliteral{q 0 0 m 0 5 l 5 10 l 10 5 l 10 0 l 7 0 l 7 5 l 3 5 l 3 0 l f
               1 j 1 J -2 5 m 5 12 l 12 5 l S Q }%
   \kern 13pt}}

\newcommand{\complementT}[1]{#1^{\complement}}

\newtheorem{theorem}{Theorem}
\newtheorem{lemma}{Lemma}
\newtheorem{definition}{Definition}

\newcommand{\foreign}{\mathcal{F}\su}
\newcommand{\native}{\mathcal{N}\su}

\newcommand{\ChannelPP}[1]{g_{\x_{#1}}}
\newcommand{\ChannelSP}[1]{g_{\y_{#1}}}
\newcommand{\ChannelPS}[1]{h_{\x_{#1}}}
\newcommand{\ChannelSS}[1]{h_{\y_{#1}}}
\newcommand{\ChannelHome}[1]{F_{ #1}}

\newcommand{\LinkPP}[1]{s_{\x_{#1}}}
\newcommand{\LinkSP}[1]{s_{\y_{#1}}}
\newcommand{\LinkPS}[1]{t_{\x_{#1}}}
\newcommand{\LinkSS}[1]{t_{\y_{#1}}}
\newcommand{\LinkHome}[1]{T_{ #1}}

\newcommand{\AnglePP}[1]{\theta_{\x_{#1}}}
\newcommand{\AnglePS}[1]{\omega_{\x_{#1}}}

\newcommand{\PRevConst}{M_\pu}
\newcommand{\SRevConst}{M_\su}

\newcommand{\PRevFun}{\mathcal{M}_\pu}
\newcommand{\SRevFun}{\mathcal{M}_\su}

\newcommand{\PLicFun}{\mathcal{P}_\pu}
\newcommand{\SLicFun}{\mathcal{P}_{\su\cu}}

\newcommand{\PLicConst}{\Pi_\pu}
\newcommand{\SLicConst}{\Pi_{\su\cu}}

\newcommand{\SPLicFun}{\mathcal{P}_{\su\pu}}
\newcommand{\SPLicConst}{\Pi_{\su\pu}}

\newcommand{\PUtFun}{\mathcal{U}_\pu}
\newcommand{\SUtFun}{\mathcal{U}_\su}
\newcommand{\CUtFun}{\mathcal{U}_\cu}

\iftoggle{SC}{\newcommand{\scf}{.55}}
{\newcommand{\scf}{.9}}
\iftoggle{SC}{\newcommand{\scft}{.7}}
{\newcommand{\scft}{.9}}

\iftoggle{SC}{\newcommand{\begs}{\begin{small}}}{\newcommand{\begs}{}}
\iftoggle{SC}{\newcommand{\ens}{\end{small}}}{\newcommand{\ens}{}}
\newcommand{\DCPbreak}{\iftoggle{SC}{}{\nonumber\\&}}
\newcommand{\DCPbreakI}{\iftoggle{SC}{}{\right.\nonumber\\&\left.}}
\newcommand{\DCPbreakII}{\iftoggle{SC}{}{\right.\right.\nonumber\\&\left.\left.}}
\newcommand{\DCPbreakIII}{\iftoggle{SC}{}{\right.\right.\right.\nonumber\\&\left.\left.\left.}}
\newcommand{\DCPbreakIV}{\iftoggle{SC}{}{\right.\right.\right.\right.\nonumber\\&\left.\left.\left.\left.}}

\newcommand{\DCspace}{\iftoggle{SC}{}{\hspace{.2in}}}

\newcommand{\insertnotationtable}{
\begin{table}[ht!]\label{notation}
	\caption{Summary of Notation}
	\begin{tabulary}{\columnwidth}{ |l | L | }\hline
		{\bf Notation} &{\bf Description}\\ \hline
		$\PPPT, \lambdaPT,\x_i$ & For the \primary operator: PPP modeling locations of BSs, BS density, location of $i\ths$ BS. \\ \hline
		$\PPPR, \lambdaPR$,\UEP & For the  \primary : PPP modeling locations of users, user density, the typical user at the origin. \\ \hline
		$\powerP{},\AntGain{\pu}{}(\cdot), \nAntP$ & For the \primary: transmit power of BSs, BS antenna pattern, and number of BS antennas. \\ \hline
		$W,\IT$ &  Licensed bandwidth, maximum interference limit for secondary operator. \\ \hline
		$\PPST, \lambdaST,\y_i$ & For the secondary operator: PPP modeling locations of BSs, BS density, location of $i\ths$ BS. \\ \hline
		$\PPSR, \lambdaSR$,\UES & For the secondary: PPP modeling locations of users, user density, the typical user at the origin. \\ \hline
		$\powerS{i},\XXi{i},\AntGain{\pu}{}(\cdot), \nAntS$ & For the secondary: transmit power of the $i\ths$ BS, Normalized transmit power of the the $i\ths$ BS defined as $\powerS{i}/\IT$ , BS antenna pattern, and number of BS antennas.\\ \hline
		$\L,\N$ &   Possible values of link type: L denotes LOS, N denotes NLOS.\\ \hline
		$\LinkPP{},\ChannelPP{}$ & For the link between \UEP~and BS at $\x$: $\LinkPP{}\in\{\L,\N\}$ denotes the link type and $\ChannelPP{}$ is the fading.\\ \hline
		$\LinkPS{},\ChannelPS{}$ & For the link between \UES~and BS at $\x$:  $\LinkPS{}\in\{\L,\N\}$ denotes the link type and $\ChannelPS{}$ is the fading.\\ \hline
		$\mathcal{H}_i,\LinkHome{i},\ChannelHome{i}$ & $\mathcal{H}_i$ is the closet (radio distance wise) primary user for the $i\ths$ secondary BS, $\LinkHome{i}$ is the type of the link between this BS and $\mathcal{H}_i$, $\ChannelHome{i}$ is the fading. \\ \hline
		$C_t$ and $\alpha_t$  & Path-loss model parameters: path-loss gain and path-loss exponent of any link of type $t\in\{\L,\N\}$. \\ \hline
		$p_\L(r),p_\N(r)$& The probability of being LOS or NLOS for a link of distance $r$.\\ \hline
		$\exclusionfun{s}{t}(x)$ & Exclusion radius for primary users of type $t$  from the secondary BS when it is associated with a $s$ type primary user located at distance $x$. \\ \hline
		$\NoiseP,\NoiseS$ & Noise power at the \UEP~and \UES. \\\hline
		$P^\mathrm{c}_\pu(\cdot),P^\mathrm{c}_\su(\cdot)$ &  Coverage probability  of \UEP~and \UES.\\ \hline
		$\Rc^{\mathrm{c}}_\pu(\cdot),\Rc^{\mathrm{c}}_\pu(\cdot)$ &  Rate coverage of \UEP~ and \UES. \\ \hline
		$\PRevFun(\cdot),\SRevFun(\cdot)$& Revenue functions for the primary and secondary operator. \\\hline
		$\PLicFun(\cdot),\SLicFun(\cdot)$& License cost functions for the primary and secondary operators to the central entity. \\\hline
		$\SPLicFun(\cdot)$& License cost functions given by the secondary operator to the primary operator. \\\hline
		$\PUtFun(\cdot),\SUtFun(\cdot),\CUtFun(\cdot)$& The total revenue  functions of the  primary operator,  the secondary operator and the central entity. \\\hline
	\end{tabulary}\vspace{-0.2in}
\end{table}
}

\allowdisplaybreaks

\begin{document}
\title{Gains of Restricted Secondary Licensing in Millimeter Wave Cellular Systems}  \author{Abhishek K. Gupta, Ahmed Alkhateeb, Jeffrey G. Andrews, and \\Robert W. Heath, Jr. \thanks{A. K. Gupta (g.kr.abhishek@utexas.edu), A. Alkhateeb (aalkhateeb@utexas.edu), J. G. Andrews (jandrews@ece.utexas.edu) and R. W. Heath Jr. (rheath@utexas.edu) are with WNCG Group, The University of Texas at Austin, Austin,
TX 78712 USA.} \thanks{ This work is supported in part by the National Science Foundation under Grant 1514275.}\thanks{A shorter version of this paper has been submitted to IEEE Globecom, April, 2016 \cite{GuptaAlkhateeb2016}.}}
\maketitle

\begin{abstract} 
	Sharing the spectrum among multiple operators seems promising in millimeter wave (mmWave) systems. One explanation is the highly directional transmission in mmWave, which reduces the interference caused by one network on the other networks sharing the same resources. In this paper, we model a  mmWave cellular system where an operator that primarily owns an exclusive-use license of a certain band can sell a restricted secondary license of the same band to another operator. This secondary network has a restriction on the maximum interference it can cause to the original network. Using stochastic geometry, we derive expressions for the coverage and rate of both networks, and establish the feasibility of secondary licensing in licensed mmWave bands. To explain economic trade-offs, we consider a revenue-pricing model for both operators in the presence of a central licensing authority. Our results show that the original operator and central network authority can benefit from secondary licensing when the maximum interference threshold is properly adjusted. This means that the original operator and central licensing authority have an incentive to permit a secondary network to restrictively share the spectrum. Our results also illustrate that the spectrum sharing gains increase with narrow beams and when the network densifies.


\end{abstract}

\section{Introduction} \label{sec:Intro}

Communication over mmWave frequencies can leverage the large bandwidth available at these frequency bands. This makes mmWave a promising candidate for next- generation cellular systems \cite{PiKhan2011,Andrews5G,Boccardi2014,Rangan2014}. Two key features of  mmWave cellular communication are directional transmission with narrow beams and sensitivity to blockage \cite{SinghBackHaul2015,Bai2014}. This results in a lower level of  interference, opening up the feasibility of spectrum sharing between multiple operators in mmWave bands \cite{GuptaAndHeath2016}. When an operator, though, already has an exclusive use of a spectral block, it will only share its spectrum if it 
results in a selfish benefit. In this paper, we establish the potential gains when a central licensing authority and a spectrum-owning operator sell a restricted-access license to a secondary operator. 

\subsection{Prior Work}

At conventional cellular frequencies, operators own exclusive licenses that give them the absolute right of using a particular frequency band. One drawback of exclusive licensing is that some portions of the spectrum remain highly underutilized \cite{FCC2002}. To overcome that, secondary network operation--also known as cognitive radio networks \cite{Haykin2005,Kang2009,Stevenson2009,Akyildiz06,Stotas2011}--can be  used \cite{Lima2012,ElSawy2014}. 
The key operational concept of secondary networks is to serve their users without exceeding a certain interference threshold at the primary network, that owns the spectrum. One main approach to guarantee that is continuous spectrum sensing \cite{Lima2012,ElSawy2014}. This, however, consumes a lot of power and time-frequency resources, which diminishes the practicality of spectrum-sensing based cognitive radio systems. As shown in \cite{SinghBackHaul2015,Bai2014}, mmWave systems experience relatively low interference due to directionality and sensitivity to blockage. This motivates sharing the mmWave spectrum among different operators without any coordination, i.e., without the licensee controlling the secondary operators \cite{GuptaAndHeath2016}. It represents, therefore, the opposite extreme versus instantaneous spectrum-sensing based cognitive radios. An intermediate solution, between these two extremes, is to allow some static coordination based on large channel statistics instead of the continuous sensing. While spectrum sharing can be beneficial for mmWave systems even without any coordination \cite{GuptaAndHeath2016}, its gain over exclusive licensing can probably be magnified with some static coordination. Exploring the potential gains of such static coordination based spectrum sharing in mmWave cellular systems is the topic of this paper.

Using stochastic geometry tools, some research has been done on analyzing the performance of cognitive radio networks at conventional cellular frequencies \cite{Khosh2013,ElSawy2014,Nguyen2012}. In \cite{Khosh2013}, a network of a primary transmitter-receiver pairs and secondary PPP users was considered, and the outage probability of the primary links were evaluated. In \cite{ElSawy2014}, a cognitive cellular network with multiple primary and secondary base stations was modeled, and the gain in the outage probability due to cognition was quantified. In \cite{Nguyen2012}, a cognitive carrier sensing protocol was proposed for a network consisting of multiple primary and secondary users, and the spectrum access probabilities were characterized. The work in \cite{Khosh2013,ElSawy2014,Nguyen2012}, though, did not consider mmWave systems and their differentiating features. In \cite{GuptaAndHeath2016}, stochastic geometry was employed to analyze spectrum-sharing in mmWave systems but with no coordination between the different operators. When some coordination exists between these operators, evaluating the network performance becomes more challenging and requires new analysis, which is one of the contributions in our work.




\subsection{Contributions} 
In this paper, we consider a downlink mmWave cellular system with a primary and a secondary operator to evaluate the benefits of secondary licensing in mmWave systems. The  main contributions of our work are summarized as follows.
\begin{itemize}
	\item \textbf{A tractable model for secondary licensing in mmWave networks:}
	We propose a model for mmWave cellular systems where an operator owns an exclusive-use license to a certain band with a provision to give a restricted license to another operator for the same band. Note that there are different ideas in the spectrum market for how this restricted license works \cite{Bae2008}.  We call the operator that originally owns the spectrum the \textit{\primary operator}, and the operator with restricted license the \textit{secondary operator}. In our model, this restricted secondary license requires the licensee to adjust the transmit power of its BSs such that the average  interference at any user of the primary operator is less than a certain threshold. Due to this restriction, the transmit power of the secondary BSs depends on the primary users in its neighborhood, and hence, it is a random variable. This required developing new analytical tools to characterize the system performance, which is one of the paper's contributions over prior work. 
	
	\item \textbf{Characterizing the performance of the restricted spectrum sharing networks:} Using stochastic geometry tools, we derive expressions for the coverage probability and area spectral efficiency of the \primary and restrcited \secondary networks as functions of the interference threshold. Results show that restricted secondary licensing can achieve coverage and rate gains for the \secondary networks with a negligible impact on the \primary network performance. Compared to the case when the secondary operator is allowed to share the spectrum without any coordination \cite{GuptaAndHeath2016}, our results show that restricted secondary licensing can increase the sum-rate of the sharing operators. This is in addition to the practical advantage of providing a way to differentiate the spectrum access of the different operators.

	\item \textbf{Optimal licensing and pricing:}  We present a revenue-pricing model for both the primary and secondary operators in the presence of a central licensing entity such as the FCC. We show that with the appropriate adjustment of the interference threshold, both the original operator and the central entity can benefit from the secondary network license. Therefore, they have a clear incentive to allow restricted secondary licensing. Further, the results show that the secondary interference threshold needs to be carefully adjusted to maximize the utility gains for the primary operator and the central licensing authority. As the optimal interference thresholds that maximize the central authority can be different than that of the primary operator, the central authority may have an incentive to push the primary operator to share even if it experiences more degradation than otherwise allowable. 
\end{itemize}

The rest of the paper is organized as follows. Section \ref{sec:Model} explains the system and network model and presents the secondary licensing rules. In Section \ref{sec:Performance}, the expressions for SINR, rate coverage probability and  aggregate median rate per unit area for each operator are derived. In Section \ref{sec:Licensing}, we explain the pricing and revenue model. Section \ref{sec:Results} presents numerical results and derives main insights of the paper. We conclude in Section \ref{sec:Conc}.

\section{Network and System Model} \label{sec:Model}
In this paper, we consider a mmWave cellular system where an operator owns an exclusive-use license to a frequency band of bandwidth $W$. There is a provision that this licensee can also give a restricted \secondary license to another operator for the same band. To distinguish the two networks, we call the 
first operator the \primary operator and the second operator as \secondary operator.

 The \primary operator has a network of BSs and users. We model the locations of the \primary BSs as a Poisson point process  (PPP) $\PPPT=\{\x_i\}$ with intensity $\lambdaPT$ and the location of users as another PPP  $\PPPR$ with intensity $\lambdaPR$. We denote the the distance of $i\ths$ primary BS from the origin by $x_i=\|\x_i\|$. Each BS of the \primary operator transmits with a power $\powerP{}$. We assume that the \secondary license allows the owning entity to use the licensed band with a restriction on the transmit power: each BS of the \secondary operator adjusts its transmit power so that its average interference on any \primary user does not exceed a fixed threshold $\IT$. We model the BS locations of the \secondary operator as a PPP $\PPST=\{\y_i\}$ with intensity $\lambdaST$ and locations of its users as another PPP $\PPSR$ with intensity $\lambdaSR$. Further, we let $y_i=\|\y_i\|$ denote the distance of the $i\ths$ secondary BS from the origin. The transmit power of the $i\ths$ \secondary BS is denoted by $\powerS{i}$. We assume that all the four PPPs are independent. The PPP assumption can be justified by the fact that nearly any BS distribution in 2D results in a small fixed SINR shift relative to the PPP \cite{Guo2015, GantiArxiv} and has been used in the past to model single and multi-operator mmWave systems \cite{Bai2014,SinghBackHaul2015,DaSilva2015,GuptaAndHeath2016}.  

\subsection{Channel and SINR Model} \label{subsec:Channel_Model}
We consider the performance of the downlink of the \primary and \secondary networks separately. In each case, we consider a typical  user to be 
located at the origin. We assume an independent blocking model where  a link between a user and a BS located at distance $r$ from this user can be  either NLOS (denoted by $\N$) with a probability  $p_\N(r)$ or a LOS (denoted by $\L$) with a probability  $p_\L(r)=1-p_\N(r)$ independent to other links.  One particular example of this model is the exponential blocking model \cite{Bai2014}, where $p_\L(r)=\exp(-\beta r)$.  The pathloss from a BS to a user is given as $\ell_t(r)=C_t r^{-\alpha_t}$ where $t\in\{\L,\N\}$ denotes the type of the BS-user link, $\alpha_t$ is the pathloss exponent, and $C_t$ is near-field gain for the $t$ type links.  

For the typical \primary user \UEP, let $\LinkPP{}$ denote the type of the link between the BS at $\x$ and this user, and let $\ChannelPP{}$ represent the channel fading. Similarly, for the typical \secondary user \UES, let $\LinkPS{i}$ and $\ChannelPS{i}$ denote the type of its link to the BS at $\x$ and its channel fading. For analytical tractability, we assume all the channels have normalized Rayleigh fading, which means that all the fading variables are exponential random variables with mean 1.

We assume that each BS is equipped with a steerable directional antenna. The BS antennas at the \primary BSs has the following radiation pattern \cite{Bai2014,Akoum2012,Hunter2008} 
\begin{align} 
G_\pu(\theta)=\begin{cases} 
\AntGain{\pu}{1} & \text{ if } |\theta|\le\theta_\mathrm{\pu\beam}/2\\
\AntGain{\pu}{2} & \text{ if } |\theta|>\theta_\mathrm{\pu\beam}/2
\end{cases},
\end{align}
where $\theta\in[-\pi,\pi]$ is the angle between the beam and the user, $\AntGain{\pu}{1}$ is the main lobe gain, $\AntGain{\pu}{2}$ is the side lobe gain, and $\theta_\mathrm{\pu\beam}$ is half-power beamwidth. To satisfy the power conservation constraint, which requires the total transmitted power to be constant and not a function of the beamwidth, we normalize the gains such that $\AntGain{\pu}{1}\frac{\theta_\mathrm{\pu\beam}}{2 \pi}+\AntGain{\pu}{2}\frac{(2 \pi-\theta_\mathrm{\pu\beam})}{2\pi}=1$. Similarly, the radiation pattern of the antennas at a \secondary BS is given by $\AntGain{\su}{}(\theta)$ with parameters $\AntGain{\su}{1}$,$\AntGain{\su}{1}$ and $\theta_\mathrm{\su\beam}$.

\insertnotationtable

Both operators follow maximum average received power based association where a user connects to a BS providing the maximum received power averaged over fading. We call this BS the {\em tagged} BS. The tagged BS steers its antenna beam towards the user to guarantee the maximum antenna gain ($\AntGain{\pu}{1}$ or $\AntGain{\su}{1}$). We take this steering direction as a reference for the other directions. We denote the angle between the antenna of a BS at $\x$ and the \primary user by $\AnglePP{}$ and the \secondary user by $\AnglePS{}$. We assume that  a user can connect only to a BS in their own network. Now, we provide the SINR expression for  the typical user of each operator (See Fig.\ref{fig:sysmod}). 

\begin{enumerate}
\item Primary user \UEP~at the origin:  Let us denote the tagged BS by $\x_0\in \Phi_1^\mathrm{T}$. The SINR for this typical user is then given as
\begin{align}
\hspace{-0.35in}\SINR_{\su0}&=\frac{\powerP{}       \AntGain{\pu}{1}\ChannelPP{0}C_{\LinkPP{0}}  x_0^{-\alpha_{\LinkPP{0}}}} 
				{
				\iftoggle{SC}{\NoiseP+{\sum\limits_{\x_i\in\PPPT\setminus\x_ 0} \powerP{}\AntGain{\pu}{}(\theta_i)\ChannelPP{i}C_{\LinkPP{i}} x_i^{-\alpha_{\LinkPP{i}}}}{
				+\sum\limits_{\y_i\in\PPST} \powerS{i}\AntGain{\su}{}(\omega_i)\ChannelSP{i}C_{\LinkSP{i}} y_i^{-\alpha_{\LinkSP{i}}}
				}
				}
				{\splitfrac{\NoiseP+\sum\limits_{\x_i\in\PPPT\setminus\x_ 0} \powerP{}\AntGain{\pu}{}(\theta_i)\ChannelPP{i}C_{\LinkPP{i}} x_i^{-\alpha_{\LinkPP{i}}}}{
				+\sum\limits_{\y_i\in\PPST} \powerS{i}\AntGain{\su}{}(\omega_i)\ChannelSP{i}C_{\LinkSP{i}} y_i^{-\alpha_{\LinkSP{i}}}
				}
				}
				} .
\end{align}
\item Secondary user \UES~at the origin:  Let us denote the tagged BS by $\y_0\in \Phi_2^\mathrm{T}$. The SINR for this typical user is then given as
\begin{align}
\hspace{-0.35in}\SINR_{\pu0}&=\frac{\powerS{0}\AntGain{\su}{1}\ChannelSS{0}C_{\LinkSS{0}} y_0^{-\alpha_{\LinkSS{0}}}} 
				{\iftoggle{SC}{\NoiseS+\sum\limits_{\x_i\in\PPPT} \powerP{}\AntGain{\pu}{}(\theta_i)\ChannelPS{i}C_{\LinkPS{i}} x_i^{-\alpha_{\LinkPS{i}}}
				+\sum\limits_{\y_i\in\PPST\setminus \y_0} \powerS{i}\AntGain{\su}{}(\omega_i)\ChannelSS{i}C_{\LinkSS{i}} y_i^{-\alpha_{\LinkSS{i}}}
				}
				{\splitfrac{\NoiseS+\sum\limits_{\x_i\in\PPPT} \powerP{}\AntGain{\pu}{}(\theta_i)\ChannelPS{i}C_{\LinkPS{i}} x_i^{-\alpha_{\LinkPS{i}}}}{
				+\sum\limits_{\y_i\in\PPST\setminus \y_0} \powerS{i}\AntGain{\su}{}(\omega_i)\ChannelSS{i}C_{\LinkSS{i}} y_i^{-\alpha_{\LinkSS{i}}}
				}}
				}\label{eq:secSINRdef}.
\end{align}
\end{enumerate}

\subsection{Restricted Secondary Licensing}
We now describe the restrictions on the \secondary licenses and the sensing mechanism used by the \secondary licensee. We assume that all \secondary BSs scan for \primary users in their neighborhood. Each \secondary BS  associates itself with the closest (radio-distance wise {\em i.e.} the one providing it the highest average received power) \primary user.   We call this associated primary user as the home \primary user  of the $i\ths$ \secondary BS and denote it by $\mathcal{H}_i$. Also, we call the \secondary BSs attached to the $i\ths$ \primary user as its native BSs (see Fig. \ref{fig:assocmod}) and denote the set of these BSs by $\mathcal{N}_i$.  

\begin{figure}[t!]
	\begin{center}
	\iftoggle{SC}{
		\includegraphics[width=0.7\textwidth,clip,trim=0 140  0 0]{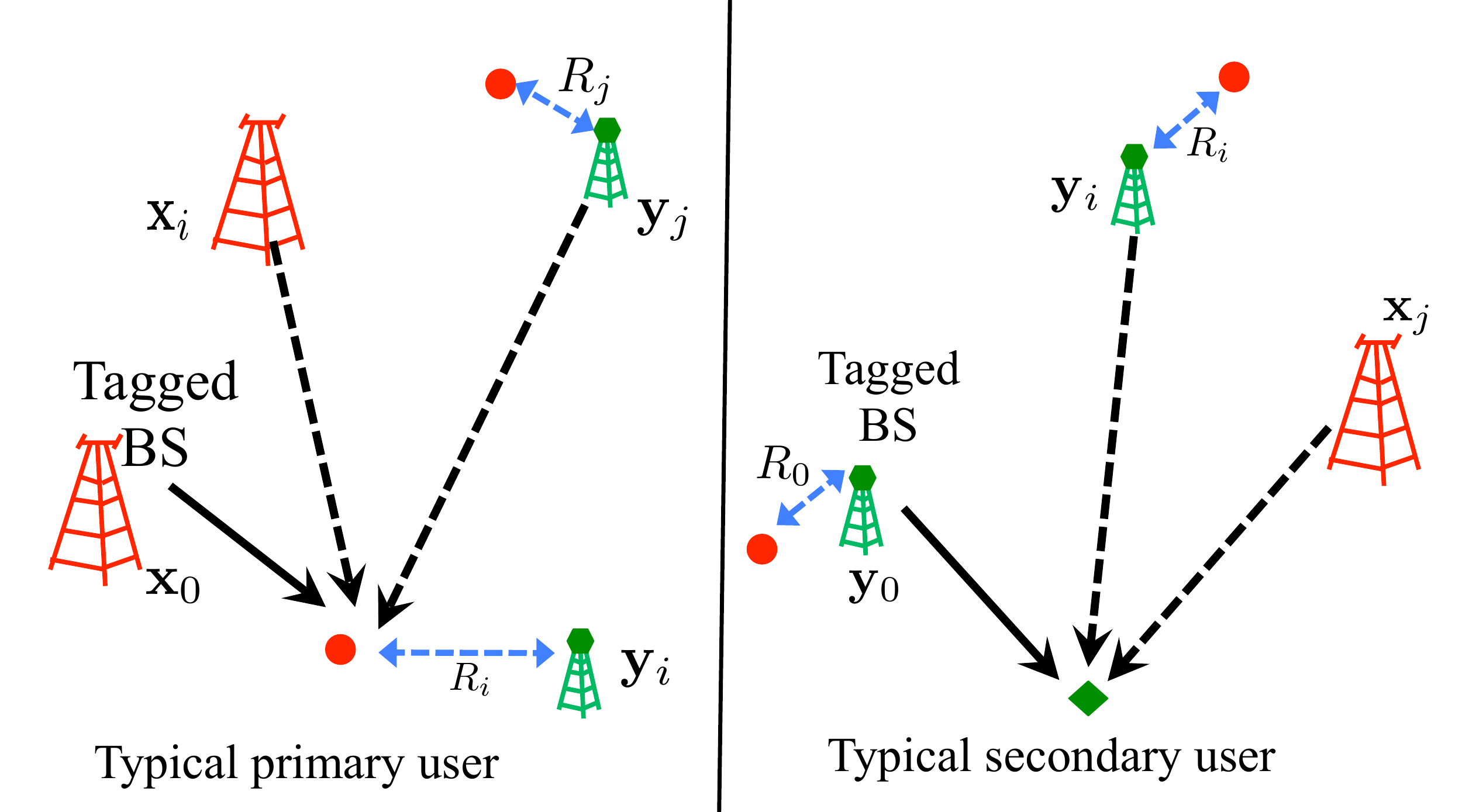}
		}{
		\includegraphics[width=0.5\textwidth,clip,trim=0 140  0 0]{figs/SysMod}
		}
		\caption{System model illustrating the SINR model for the typical \primary and the \secondary user.  $i\ths$ \secondary BS is attached to the closest \primary user where  distance between the two is denoted by $R_i$. }
		\label{fig:sysmod}
	\end{center}
\end{figure}

\begin{figure}
\begin{center}
\iftoggle{SC}{
		\fbox{\includegraphics[width=0.5\textwidth,clip,trim=50 20  20 0]{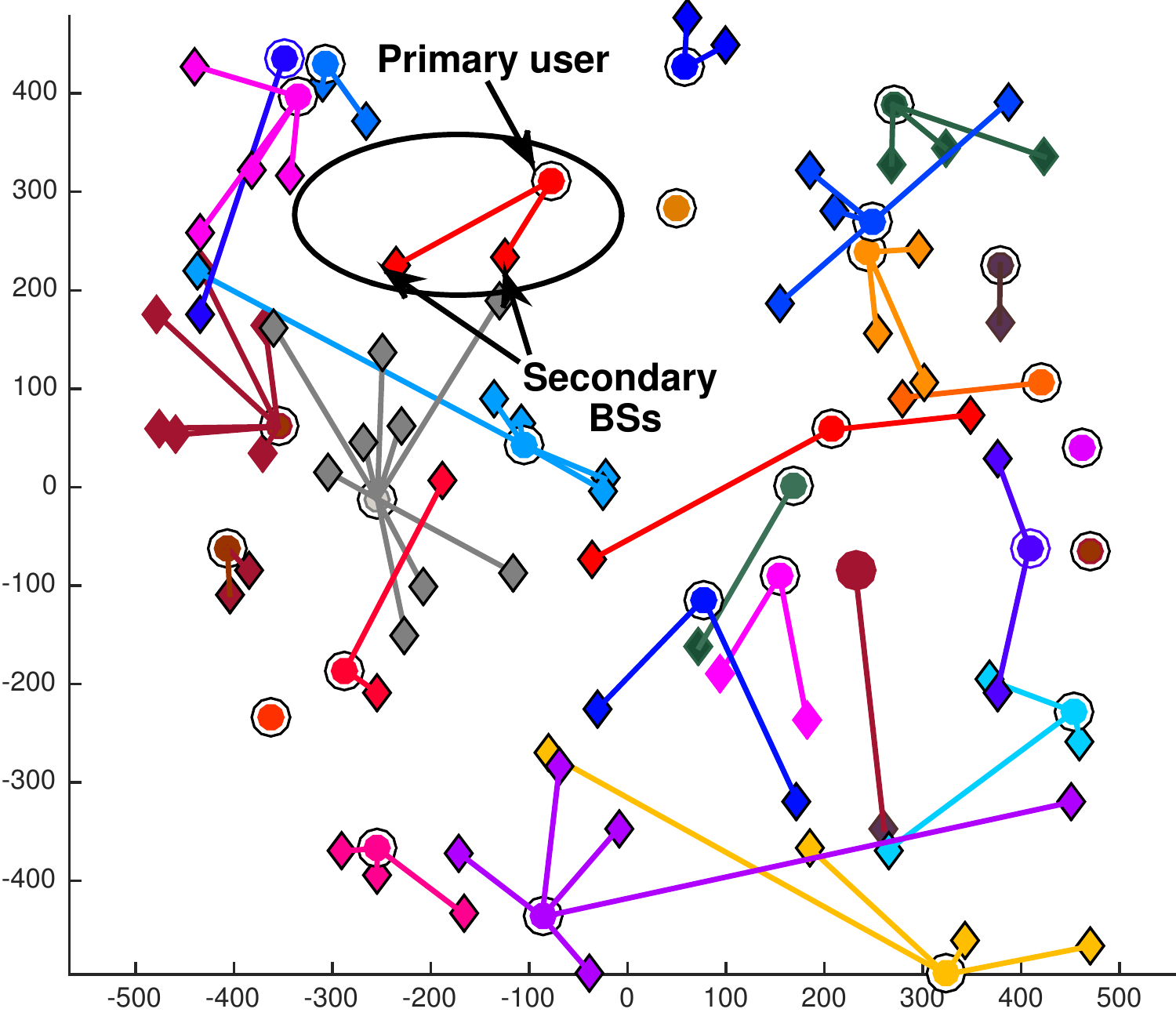}}}{
		\fbox{\includegraphics[width=0.4\textwidth,clip,trim=50 20  20 0]{figs/AssociationPtoS}}}
\caption{The association of \secondary BSs (diamonds) to their home \primary user (circles) in a particular realization of the adopted mmWave system. Each \secondary BS adjusts its transmit power to keep the interference on its home \primary user less than a given limit $\IT$.}
\label{fig:assocmod}
\end{center}
\end{figure}
Let us denote the distance between $i\ths$ \secondary BS and its home primary user $\mathcal{H}_i$  by $R_i$, 
and the type of the link between them by $\LinkHome{i}$. Note that $R_i$ for a secondary BSs  is not independent of $R_i$'s of its adjacent secondary BSs. However, for tractability, we assume that $R_i$ and $T_i$ are independent over $i$, which is a standard assumption in modeling similar association of the interfering mobile transmitters to their respective BSs in uplink analysis \cite{AndGupDhi2016}. 
For a given $R_i=r$ and $T_i=T$, all the other primary users will be outside certain exclusion regions, which is different for the LOS and NLOS users. For a primary user of link type $t$, the radius of its exclusion region, denoted by $\exclusionfun{T}{t}(r)$, is given as 
\begin{align}
\exclusionfun{T}{t}(r)&=\left(\fracS{C_t}{C_T}\right)^\frac{1}{\alpha_t}r^\frac{\alpha_T}{\alpha_t}.
 \end{align}
 Now, the joint distribution of $R_i$ and $T_i$ is given as follows:
 \begin{align}
&f_{R_i}(r,\LinkHome{i}=T)=\DCPbreak2\pi \lambdaPR p_T(r)r\exp\left(-\lambdaPR\left( V_\L(\exclusionfun{T}{\L}(r))+ V_\N(\exclusionfun{T}{\N}(r))\right)\right),
\end{align}
where 
$V_T(r)$ denotes the volume for LOS or NLOS and is defined as
$V_T(r)=2\pi \int_0^r p_T(r)rdr$. See Appendix \ref{app:RiDist} for the derivation of this distribution. Note that $\exclusionfun{T}{t}(r)=r$ when $T=t$.

Recall that the $i\ths$ \secondary BS is restricted to transmit at a certain power such that the average interference at the home user, which is equal to $\powerS{i} C_{T_i}/R_i^{\alpha_{T_i}}$m is below a threshold $\IT$. Therefore, its  transmit power  is given by   
\begin{align} \label{eq:Sec_Power}
\powerS{i}=\IT R_i^{\alpha_{T_i}}/C_{T_i}.
\end{align}
The joint distribution of $\powerS{}$ and $\LinkHome{i}$ can be computed using  transformation of variables as
\begin{align*}
&f_{\powerS{i}}(p,T_i=T)
=\frac{2\pi \lambdaPR a^2}{\alpha_{T}} p^{\frac{2}{\alpha_{T}}-1} p_{T}\left(a p^\frac1{\alpha_{T}}\right)
\DCPbreak\DCspace\DCspace\iftoggle{SC}{}{\times}
\exp\left(-\lambdaPR V_{T}\left(a p^\frac1{\alpha_{T}}\right)-\lambdaPR V_{\complementT{T}}\left(\exclusionfun{T}{\complementT{T}}\left(a p^\frac1{\alpha_{T}}\right)\right)\right),
\end{align*}
where  $a=(C_T/\IT)^{1/\alpha_{T}}$ and $\complementT{T}$ denotes the complement of $T$ {\em i.e.}  $\complementT{T}=\L$ if $T=\N$ and vice-versa.

Finally, if $\XXi{i}=\powerS{i}/\IT$ denotes the normalized transmit power of this BS, then  $\XXi{i}= R_i^{\alpha_{T_i}}/C_{T_i}$, which is a random variable  independent of $\IT$.

\section{Performance Analysis} \label{sec:Performance}
One of the important metrics to quantify the performance of a cellular system is the coverage probability. It is defined as the probability that the SINR at a typical user from its associated BS is above a threshold $\SThres$,
\begin{equation}
\mathrm{P}^\mathrm{c}(\SThres)=\prob{\SINR>\SThres}.
\end{equation}
In this section, we compute the  coverage probability of the typical users \UES~and \UEP.

\subsection{Coverage Probability of the \Secondary Operator} \label{subsec:SecondaryCov}
From the perspective of \UES, the \secondary BSs can be divided into two independent PPPs: LOS BSs ${\PPST}_\L$ and NLOS BSs ${\PPST}_\N$ based on the link type $\LinkSS{i}$ of each BS. Similarly, the \primary BSs are divided into LOS BSs ${\PPPT}_\L$ and NLOS BSs ${\PPPT}_\N$. Recall that we adopt a maximum average received power based association, in which any \secondary user  will associate with the BS  providing highest average received power.  Since each BS has a different transmit power $\powerS{i}$, the BS association to the typical user will be affected by this transmit power. Let $\laplace{I}(s)$ denote the Laplace transform of interference $I$. Now, we give the coverage probability of \UES~in the following Lemma.
 
\begin{lemma}\label{lemma:SINRCov1}
The coverage probability of a typical \secondary user is given as \iftoggle{SC}{}{$\Pc^\mathrm{c}_\mathrm{S}(\SThres)=$}
\begin{align}
\iftoggle{SC}{\Pc^\mathrm{c}_\mathrm{S}(\SThres)=}{}&
\sum_{t_0\in\{\L,\N\}}
\int_0^\infty2\pi\lambdaST
\imdfuni{t_0}{u}
\exp\left(-\frac{\SThres u^{\alpha_{t_0}} }{\IT \AntGain{\su}{1} }\NoiseS\right)
\laplace{I_\pu}\left(\frac{\SThres u^{\alpha_{t_0}} }{ \IT\AntGain{\su}{1}  }\right)
\DCPbreak\DCspace
\laplace{I'_\su}\left(\frac{\SThres u^{\alpha_{t_0}} }{ \IT\AntGain{\su}{1} }\right)
\iftoggle{SC}{\nonumber\\
&}{}
\exp\left(-2\pi\lambdaST\int_0^{u^{\fracS{\alpha_{t_0}}{\alpha_\L}}}
\imdfuni{\L}{z}
z\intd z
\DCPbreakI\DCspace\DCspace
-2\pi\lambdaST\int_0^{u^{\fracS{\alpha_{t_0}}{\alpha_\N}}}
\imdfuni{\N}{z}
z\intd z\right)
u\intd u
\end{align}
where  $I_\pu$ is the interference from the \primary BSs, $I'_\su$ is the interference from the \secondary BSs satisfying $u^{-\alpha_{\LinkSS{0}}}>\XX C_{\LinkSS{j}} \y_j^{-\alpha_{\LinkSS{j}}}$ and $\imdfuni{t}{u}$ is defined as
\begin{align}
\imdfuni{t}{u}&=\mathbb{E}_{\XX}\left[ p_{t}\left(u\XX^{\frac1{\alpha_{t}}}C_{t}^{\frac{1}{\alpha_{t}}}\right)\XX^{\frac{2}{\alpha_{t}}}\right]C_t^{\frac{2}{\alpha_{t}}}.
\end{align}
\end{lemma}
\begin{proof}
See Appendix \ref{app:SSINRCoverageStep1}.
\end{proof}
This result is interesting because the distribution of the secondary transmit power $\powerS{}$ is decoupled from most of the terms, which noticeably simplifies the final expressions. As seen from \eqref{eq:secSINRdef}, the term $\powerS{}$  is present in the association rule, serving power, and the interference. In Lemma \ref{lemma:SINRCov1}, this dependency of the coverage probability on $\powerS{}$ is reduced to only one function $K_t(\cdot)$ (see Appendix \ref{app:SSINRCoverageStep1} for the techniques used), making the whole integral easily computable, which is a key analytical contribution of the paper.  
Now, we derive the Laplace transforms of $I_\pu$ and $I'_\su$ which are given in the following Lemmas.

\begin{lemma}\label{lemma:SecPerSI}
The Laplace transform of the interference $I'_\su$ from the \secondary  network is given as
\iftoggle{SC}{
\begin{align}
&\laplace{I'_{\su}}(s)=
\exp\left(-\lambdaST  \sum_{k=1}^2a_k
F_{\su}(s \IT \AntGain{\su}{k} ,u^{{\alpha_{t_0}}})
\right)\\
\text{where }a_1=\theta_{\su\beam}/(2\pi),\text{ }&a_2=1-a_1 
\text{ and } F_{\su}(B,e)=
2\pi\sum_{t\in\{\L,\N\}}{\int_{e^{\frac1{\alpha_t}}}^\infty \frac{\imdfuni{t}{v}}{1+B^{-1}v^{\alpha_t}} v\intd v}.\hspace{0.9in}\nonumber
\end{align}}
{
\begin{align}
&\hspace{0.4in}\laplace{I'_{\su}}(s)=
\exp\left(-\lambdaST  \sum_{k=1}^2a_k
F_{\su}(s \IT \AntGain{\su}{k} ,u^{{\alpha_{t_0}}})
\right)\\
&\text{where }a_1=\theta_{\su\beam}/(2\pi),\text{ }, a_2=1-a_1 
\iftoggle{SC}{}{\nonumber\\}
& \text{and } F_{\su}(B,e)=
2\pi\sum_{t\in\{\L,\N\}}{\int_{e^{\frac1{\alpha_t}}}^\infty \frac{\imdfuni{t}{v}}{1+B^{-1}v^{\alpha_t}} v\intd v}.\hspace{1.1in}\nonumber
\end{align}
}
\end{lemma}
\begin{proof}
See Appendix \ref{app:SecPerSI}.
\end{proof}

\begin{lemma}\label{lemma:SecPerPI}
The Laplace transform of the interference $I_\pu$ from the \primary network is given as
\begin{align}
\laplace{I_\pu}(s)&=\exp\left(-\lambdaPT  \sum_{k=1}^2b_k
F_\pu(s \AntGain{\pu}{k})
\right)
\end{align}
where $b_1=\theta_{\pu\beam}/(2\pi),b_2=1-b_1$,  
\begin{align*}
F_{\pu}(B)&=2\pi\sum_{t\in\{\L,\N\}}\int_{0}^\infty \frac{ \imdfunP{t}{u}}{1+B^{-1}v^{\alpha_t}}
v\intd v, \text{ and }& 
\iftoggle{SC}{}{\nonumber\\}
\imdfunP{t}{u}&= p_{t}\left(u\powerP{}^{\frac1{\alpha_{t}}}C_{t}^{\frac{1}{\alpha_{t}}}\right)\powerP{}^{\frac{2}{\alpha_{t}}}C_\L^{\frac{2}{\alpha_{t}}}.
\end{align*}
\end{lemma}
\begin{proof}
See Appendix \ref{app:SecPerPI}.
\end{proof}
Now, substituting the Laplace transform of the \primary and \secondary interference at the \UES in Lemma \ref{lemma:SINRCov1}, we can compute the final expression of the  coverage probability, which we give in the following theorem.
\begin{theorem}\label{thm:PcS}
The coverage probability of a  typical user of the \secondary operator in a mmWave system with  restricted \secondary licensing  is given as
\begin{align}
&\Pc^\mathrm{c}_\mathrm{S}(\SThres)=\sum_{t_0\in\{\L,\N\}}
\int_0^\infty2\pi\lambdaST
\exp\left[-\lambdaPT  \sum_{k=1}^2b_k
	F_\pu\left(\frac{\SThres u^{\alpha_{t_0}}\AntGain{\pu}{k} }{\IT\AntGain{\su}{1}  } \right)
	\DCPbreakI\DCspace
	-\lambdaST  \sum_{k=1}^2a_k
	F_{\su}\left(\frac{\SThres u^{\alpha_{t_0}} \AntGain{\su}{k}}{ \AntGain{\su}{1} }   ,u^{{\alpha_{t_0}}}\right)
	\right.\nonumber\\
&\DCspace\left.	-\frac{\SThres u^{\alpha_{t_0}} }{\IT\AntGain{\su}{1} }\NoiseS
-2\pi\lambdaST\int_0^{u^{\fracS{\alpha_{t_0}}{\alpha_\L}}}
\imdfuni{\L}{z}
z\intd z
\DCPbreakI\DCspace\DCspace
-2\pi\lambdaST\int_0^{u^{\fracS{\alpha_{t_0}}{\alpha_\N}}}
\imdfuni{\N}{z}
z\intd z\right]
K_{t_0}(u)u\intd u.
\end{align}
\end{theorem}

Since the expression in Theorem \ref{thm:PcS} is complicated, we consider the following three special cases to give simple and closed form expressions.

\noindent\textbf{Special Cases:}
\begin{enumerate}
\item[(i)] Consider a mmWave network with identical parameters for LOS and NLOS channels (which is also the typical assumption for a  UHF system). For this case, we can combine the LOS and NLOS PPPs in to a single PPP of type $t$ for which $\imdfuni{t}{u}$ is given as
 \begin{align}
 \imdfuni{t}{u}=\mathbb{E}_{\XX}\left[\XX^{\frac{2}{\alpha}}\right]C^{\frac{2}{\alpha}}={\expect{R^2}}=\left({\lambdaPR\pi}\right)^{-1},
 \end{align}
  which is no longer a function of $u$. Let us denote this constant by $K$. Similarly, $\imdfunP{t}{u}$ can be simplified as $\imdfunP{t}{u}=\powerP{}^{\frac{2}{\alpha}}C^{\frac{2}{\alpha}}$, which is no longer a function of $u$ and hence can be denoted by $M$. Therefore, $F_{\su}(B,e)$ and $F_{\pu}(B/\IT)$ can be simplified as follows:
\begin{align*}
F_{\su}(B,e)&=K\int_{(e)^{\frac1\alpha}}^\infty \frac{1}{1+(B)^{-1}v^{\alpha}}2\pi v\intd v
\DCPbreak=
KB^{\frac2{\alpha}}\int_{(e/B)^{\frac1\alpha}}^\infty \frac{1}{1+v^{\alpha}}2\pi v\intd v,\\
F_{\pu}(B/\IT)&=\int_{0}^\infty \frac{2\pi\powerP{}^{\frac{2}{\alpha}}C_t^{\frac{2}{\alpha}}}{1+\IT B^{-1}v^{\alpha}}
 v\intd v\DCPbreak
 =(\IT B^{-1})^{-\frac2\alpha}\int_{0}^\infty \frac{2\pi\powerP{}^{\frac{2}{\alpha}}C^{\frac{2}{\alpha}}}{1+ {v}^{\alpha}}
v\intd v
.
\end{align*}
\begin{align*}
&\text{Now, Let us define }\rho(\alpha,\tau)= \int_{\tau^{-1/\alpha}}^\infty\frac{1}{1+v^{\alpha}}2v\intd v 
\text{ and } \DCPbreak\rho(\alpha)= \rho(\alpha,\infty)\text{, then } \hspace{1.85in} \\
&F_{\su}(B,e)=\pi KB^{\frac2{\alpha}}\rho(\alpha,B/e)\Rightarrow\DCPbreak
 F_{\su}\left(\frac{\SThres u^{\alpha} }{ \AntGain{\su}{1} }  \AntGain{\su}{k} ,u^{{\alpha}}\right)
=\pi K\SThres^{\frac2{\alpha}}u^2 {\left(\frac{\AntGain{\su}{k}}{ \AntGain{\su}{1} }\right)}^{\frac2{\alpha}}\rho\left(\alpha,\frac{\SThres \AntGain{\su}{k} }{ \AntGain{\su}{1} }  \right),
\\
&F_{\pu}\left(\frac{B}\IT\right)=\pi \left(\frac{B\powerP{} C}{\IT}\right)^{\frac{2}{\alpha}}\rho(\alpha)
\Rightarrow \DCPbreak
F_{\pu}\left(\frac{\SThres u^{\alpha}\AntGain{\pu}{k} }{\IT\AntGain{\su}{1}  } \right)
=\pi 
\left(\frac{\powerP{}\SThres C}{\IT}\frac{\AntGain{\pu}{k}}{\AntGain{\su}{1}  } \right)^{\frac{2}{\alpha}}  
u^2  
\rho(\alpha) .
\end{align*}
Let us define
$F''_\pu=\sum_{k=1}^2b_k {\left(\frac{\AntGain{\pu}{k}}{\AntGain{\su}{1}  } \right)}^{\frac{2}{\alpha}}\rho(\alpha)$
and
	$F''_\su(\SThres)=\sum_{k=1}^2 a_k {\left(\frac{\AntGain{\su}{k}}{ \AntGain{\su}{1} }\right)}^{\frac2{\alpha}}\rho\left(\alpha,\frac{\SThres \AntGain{\su}{k} }{ \AntGain{\su}{1} }  \right)$.
Then, the  coverage probability is given as
\iftoggle{SC}{\begin{align*}
\Pc^\mathrm{c}_\mathrm{S}(\SThres)&=2\pi\lambdaST
K
\int_0^\infty\expU{-\frac{\SThres \NoiseS }{{{\IT}}\AntGain{\su}{1} }{u}^{\alpha}-u^2 \left(\pi\lambdaPT \IT^{-\frac2\alpha}\SThres^{\frac2\alpha}(\powerP{} C)^{\frac{2}{\alpha}}  F''_\pu+\pi \lambdaST K \SThres^{\frac2{\alpha}} F''_\su(\SThres)+\pi\lambdaST 
K\right)}
u\intd u.
\end{align*} }
{\begin{align*}
&\Pc^\mathrm{c}_\mathrm{S}(\SThres)=2\pi\lambdaST
K
\int_0^\infty\expU{-\frac{\SThres \NoiseS }{{{\IT}}\AntGain{\su}{1} }{u}^{\alpha}}\nonumber\\&\expU{-u^2 \left(\pi\lambdaPT \IT^{-\frac2\alpha}\SThres^{\frac2\alpha}(\powerP{} C)^{\frac{2}{\alpha}}  F''_\pu+\pi \lambdaST K \SThres^{\frac2{\alpha}} F''_\su(\SThres)+\pi\lambdaST 
K\right)}
u\intd u.
\end{align*} }
\item[(ii)] Consider a mmWave system with identical LOS and NLOS channels in the interference limited scenario. In this case, the  coverage probability is given as $\Pc^\mathrm{c}_\mathrm{S}(\SThres)=$
\begin{align} 
&
\hspace{-0.2in}\left[
1+\SThres^{\frac2\alpha}\left(
\frac{\lambdaPT }
{\IT^{\frac2\alpha} \lambdaST} 
\frac{(\powerP{} C)^{\frac{2}{\alpha}}}
{{(\lambdaPR\pi)}^{-1}} 
 \sum_{k=1}^2b_k {\left(\frac{\AntGain{\pu}{k}}{\AntGain{\su}{1}  } \right)}^{\frac{2}{\alpha}}\rho(\alpha)
 \DCPbreakII +
   \sum_{k=1}^2 a_k {\left(\frac{\AntGain{\su}{k}}{ \AntGain{\su}{1} }\right)}^{\frac2{\alpha}}\rho\left(\alpha,\frac{\SThres \AntGain{\su}{k} }{ \AntGain{\su}{1} }  \right)\right)
 \right]^{-1}\label{eq:simplifiedSec}
\end{align}

\textbf{Impact of secondary densification and $\IT$:}
We can see from the result in \eqref{eq:simplifiedSec} that $\Pc^\mathrm{c}_\mathrm{S}(\SThres)$ is invariant if the term $\IT\lambdaST^{\alpha/2}$ is kept constant. This is due to the following observation: if we increase the secondary BS density $\lambdaST$ by a factor of $a$, the distance of  the secondary BSs  decreases by a factor of $\sqrt{a}$ and therefore, the secondary interference increases by the factor of  ${a}^{\alpha/2}$, and if we increase the interference limit $\IT$ by $a$, the secondary interference also increases by $a$. Therefore, densifying the secondary network while reducing its interference threshold by the appropriate ratio keeps the coverage probability constant.

\textbf{Impact of narrowing \secondary antenna:} If we assume that the \secondary antennas are uniform linear arrays of $\nAntS$ antennas, then $a_k$'s and $\AntGain{\su}{k}$'s can be approximated as \cite{VanTrees2002}
\begin{align}
&a_1=\frac{\kappa}\nAntS, \AntGain{\su}{1}=\nAntS, a_2=1-\frac{\kappa}\nAntS,\text{ and }\DCPbreak
\AntGain{\su}{2}=\frac{(1-\kappa)\nAntS}{\nAntS-\kappa}(\approx 1-\kappa\text{ for large } \nAntS)\nonumber
\end{align}
where $\kappa$ is some constant. Now, the term denoting \primary interference decreases as $\frac{\kappa}{\nAntS^{2/\alpha}}$ and the term denoting the \secondary interference decreases as $\frac{\kappa}N\rho(\alpha,\SThres)+\left(\frac{1-\kappa}{\nAntS}\right)^{2/\alpha}$ $\rho\left(\alpha,\SThres\frac{1-\kappa}{\nAntS}\right)$. Therefore, narrowing the \secondary antennas beamwidth noticeably improves the \secondary performance.

\textbf{Impact of narrowing \primary antenna:} With a similar assumption for the primary BSs to have uniform linear arrays of $\nAntP$ antennas, $b_k$'s and $\AntGain{\pu}{k}$'s can be approximated as $
b_1=\frac{\kappa}\nAntP, \AntGain{\pu}{1}=\nAntP, b_2=1-\frac{\kappa}\nAntP,\text{ and }\AntGain{\pu}{2}=\frac{(1-\kappa)\nAntP}{\nAntP-\kappa}(\approx 1-\kappa\text{ for large } N)$. 
Now, the term denoting the \primary interference decreases as $\frac{\kappa}{\nAntP^{1-2/\alpha}}+\left({1-\kappa}\right)^{2/\alpha}$ while the term denoting the \secondary interference remains constant. Therefore, narrowing the \primary beamwidth slightly improves the \secondary performance. For high value of $\IT$ where the \secondary interference dominates, the \secondary performance does not improve by narrowing the \primary antennas. 

\item[(iii)] Suppose that both operators have the same beam patterns with zero side-lobe gain. In this case, the coverage probability can be simplified to a closed form expression: 
\iftoggle{SC}{
\begin{align*}
&\hspace{1in}\Pc^\mathrm{c}_\mathrm{S}(\SThres)=
{\left[1+\frac{\theta_b}{2\pi}\SThres^{2/\alpha}\left(\frac{\lambdaPT}{\lambdaST} \IT^{-\frac2\alpha}\frac{\powerP{}^{\frac{2}{\alpha}}\rho(\alpha)}{(\lambdaPR\pi C^{\frac2\alpha})^{-1}}+  \rho(\alpha,\SThres)\right)\right]}^{-1}\\
&\text{which, for $\alpha=4$, becomes }\Pc^\mathrm{c}_\mathrm{S}(\SThres)=
\left[{1+\frac{\theta_b\sqrt\SThres}{2\pi}\left(\frac1{\sqrt{\IT}} \frac{\lambdaPT}{\lambdaST} \frac{\sqrt{C\powerP{}}}{(\lambdaPR)^{-1}}\frac{\pi^2}2 + \tan^{-1}(\sqrt\tau)\right)}\right]^{-1}.
\end{align*}}
{
\begin{align*}
&\Pc^\mathrm{c}_\mathrm{S}(\SThres)\nonumber\\=&
{\left[1+\frac{\theta_b}{2\pi}\SThres^{2/\alpha}\left(\frac{\lambdaPT}{\lambdaST} \IT^{-\frac2\alpha}\frac{\powerP{}^{\frac{2}{\alpha}}\rho(\alpha)}{(\lambdaPR\pi C^{\frac2\alpha})^{-1}}+  \rho(\alpha,\SThres)\right)\right]}^{-1}
\end{align*}
which, for $\alpha=4$, becomes $\Pc^\mathrm{c}_\mathrm{S}(\SThres)=$
\begin{align*}
&\left[{1+\frac{\theta_b\sqrt\SThres}{2\pi}\left(\frac1{\sqrt{\IT}} \frac{\lambdaPT}{\lambdaST} \frac{\sqrt{C\powerP{}}}{(\lambdaPR)^{-1}}\frac{\pi^2}2 + \tan^{-1}(\sqrt\tau)\right)}\right]^{-1}.
\end{align*}
}
\end{enumerate}

\subsection{Coverage Probability of the \Primary Operator} \label{subsec:PrimaryCov}
Similar to the \secondary case, for \UEP~also, all the \primary and \secondary BSs can be divided into two independent LOS and NLOS PPPs based on the link type between each BS and \UEP. Recall that we have assumed maximum average received power based association, in which any \primary user  will associate with the BS  $\x_0$  providing highest average received power. We, now compute the  coverage probability of the typical \primary user which is given in Lemma \ref{lemma:PSINRCoverageStep1}.
\begin{lemma}\label{lemma:PSINRCoverageStep1}
The  coverage probability of the \primary operator is given as $\iftoggle{SC}{}{P^\mathrm{c}_\pu(\SThres)=}$
\begin{align}
\iftoggle{SC}{P^\mathrm{c}_\pu(\SThres)=}{}&\sum_{t_0\in\{\L,\N\}}
\int_0^\infty2\pi\lambdaPT
\imdfunP{t_0}{u}
\exp\left(-\frac{\SThres u^{\alpha_{t_0}} }{ \AntGain{\pu}{1} }\NoiseP\right)
\laplace{I'_\pu}\left(\frac{\SThres u^{\alpha_{t_0}} }{ \AntGain{\pu}{1}  }\right)
\DCPbreak\DCspace
\laplace{I_\su}\left(\frac{\SThres u^{\alpha_{t_0}} }{ \AntGain{\pu}{1} }\right)
\iftoggle{SC}{\nonumber\\&}{}
\exp\left(-2\pi\lambdaPT\int_0^{u^{\fracS{\alpha_{t_0}}{\alpha_\L}}}
\imdfunP{\L}{z}
z\intd z\DCPbreakI\DCspace
-2\pi\lambdaPT\int_0^{u^{\fracS{\alpha_{t_0}}{\alpha_\N}}}
\imdfunP{\N}{z}
z\intd z\right)
u\intd u
\end{align}
where $I'_\pu$ is the interference from the \primary operator conditioned on the fact that the serving BS is at $\x_0$ and  $I'_\su$ is  the interference from the \secondary operator.
\end{lemma}
\begin{proof}
The proof is similar to the proof in Appendix \ref{app:SSINRCoverageStep1}. The only difference is that the computations for the \primary and \secondary operators are interchanged and there will not be any expectation with respect to the transmit power of the \primary BSs as the \primary transmit power is deterministic.
\end{proof}
We now compute the Laplace transforms of the \primary and \secondary interference which is given in the following two Lemmas.

\begin{lemma}
The Laplace transform of the interference $I'_\pu$ from the conditioned \primary network is given as 
\iftoggle{SC}{
\begin{align}
\laplace{I'_{\pu}}(s)&=\exp\left(-\lambdaPT  \sum_{k=1}^2b_k
E_{\pu}(s  \AntGain{\pu}{k} ,u^{{\alpha_{t_0}}})
\right)\\
\text{where $E_{\pu }(B,e)$ is given as:    }E_{\pu}(B,e)&=B^{\frac2{\alpha_t}}\sum_{t\in\{\L,\N\}}\int_{(e/B)^{\frac1{\alpha_t}}}^\infty \frac{1}{1+v^{\alpha_t}}\imdfunP{t}{vB^{\frac2{\alpha_t}}}2\pi v\intd v.\hspace{1.6in}\nonumber
\end{align}}{
\begin{align}
\laplace{I'_{\pu}}(s)&=\exp\left(-\lambdaPT  \sum_{k=1}^2b_k
E_{\pu}(s  \AntGain{\pu}{k} ,u^{{\alpha_{t_0}}})
\right)\end{align}
where $E_{\pu }(B,e)$ is given as:    
\begin{align}
E_{\pu}(B,e)&=B^{\frac2{\alpha_t}}\sum_{t\in\{\L,\N\}}\int_{(e/B)^{\frac1{\alpha_t}}}^\infty \frac{1}{1+v^{\alpha_t}}\imdfunP{t}{vB^{\frac2{\alpha_t}}}2\pi v\intd v.\nonumber
\end{align}}
\end{lemma}
\begin{proof}
The proof is similar to the proof of Lemma \ref{lemma:SecPerSI} with  only difference being lack of any expectation with respect to the primary BSs' transmit powers.
\end{proof}
The \secondary interference can be written as sum of following two interferences: the interference $I_{\foreign}$ from the BSs that are not in native set $\mathcal{N}_0$ and interference $I_{\native}$ from the BSs that are in $\mathcal{N}_0$.  The following Lemma gives the Laplace transform of the interference from the secondary operator where functions $E_{\foreign}$ and $E_{\native}$ are due to  $I_{\foreign}$ and  $I_{\native}$ respectively.
\begin{lemma}\label{lemma:PriSI}
The Laplace transform of the interference from the \secondary network  is given as 
\begin{align}
\laplace{I_{\su}}(s)&=
		\exp\left( -\lambdaST  \sum_{k=1}^2a_k \left(
E_{\foreign}\left(s  \AntGain{\su}{k} ,\IT\right) +E_{\native}(s\AntGain{\su}{k})\right)
\right)
\end{align}
where $E_{\foreign }(B,\IT)$  and $E_{\native}(B)$ are given as
\begin{align}
E_{\foreign}(B,\IT)
&=(B\IT)^{\frac2{\alpha_t}}\sum_{t\in\{\L,\N\}}\int_{(\IT B)^{-\frac1{\alpha_t}}}^\infty \frac{\imdfuni{t}{v (B\IT)^{\frac1{\alpha_t}}}}{1+v^{\alpha_t}}2\pi v\intd v\\
E_{\native}(B)&=
		\frac{1}{1+
				(B\IT )^{-1}}\int_0^{1}
			(\imdfuni{\L}{v}+
			\imdfuni{\N}{v})2\pi v \intd v
\end{align}
\end{lemma}
\begin{proof}
See Appendix \ref{app:PriSI}.
\end{proof}

Now, substituting the Laplace transforms of $I'_\pu$ and $I_\su$ in  Lemma \ref{lemma:PSINRCoverageStep1}, we can compute the final expression of the coverage probability, which is given in the following theorem.

\begin{theorem}\label{thm:PPc}
The coverage probability of a typical  user of the \primary  operator in a mmWave system  with \secondary licensing is given as
\begin{align}
\iftoggle{SC}{}{&}\Pc^\mathrm{c}_\pu(\SThres)=\iftoggle{SC}{&}{}\sum_{t_0\in\{\L,\N\}}
\int_0^\infty2\pi\lambdaPT
\imdfunP{t_0}{u}
\expU{-\frac{\SThres u^{\alpha_{t_0}} }{ \AntGain{\pu}{1} }\NoiseS}
\DCPbreak
\expU{-\lambdaPT  \sum_{k=1}^2b_k
E_{\pu}\left(\frac{\SThres u^{\alpha_{t_0}} }{ \AntGain{\pu}{1} }  \AntGain{\pu}{k} ,u^{{\alpha_{t_0}}}\right)
-\lambdaST  \sum_{k=1}^2a_k
E_{\foreign}\left(\frac{\SThres u^{\alpha_{t_0}} }{ \AntGain{\pu}{1} } \AntGain{\su}{k} ,\IT\right)}\nonumber\\\iftoggle{SC}{\times}{} &
\expU{-\lambdaST \sum_{k=1}^2 a_k E_{\native}\left(\frac{\SThres u^{\alpha_{t_0}} }{ \AntGain{\pu}{1} }\AntGain{\su}{k}\right)}
\DCPbreak
\expU{\left(-2\pi\lambdaPT\int_0^{u^{\fracS{\alpha_{t_0}}{\alpha_\L}}}
\imdfunP{\L}{z}
z\intd z
-2\pi\lambdaPT\int_0^{u^{\fracS{\alpha_{t_0}}{\alpha_\N}}}
\imdfunP{\N}{z}
z\intd z\right)}
u\intd u.
\end{align}
\end{theorem}

\noindent\textbf{Special Cases:} Similar to the secondary case, consider a mmWave network with identical parameters for LOS and NLOS channels. For this case, $\imdfunP{t}{u}$ and $\imdfuni{t}{u}$ are replaced by constant $M$ and $K$.
Now, $E_{\pu}(B,e)$, $E_{\foreign}(B)$ and $E_{\foreign}(B,\IT)$ can be simplified as follows:
\iftoggle{SC}{
\begin{align*}
\iftoggle{SC}{}{&}E_{\pu}(B,e)=
\pi M B^{\frac2{\alpha}}\rho(\alpha,B/e),\text{ }\DCPbreak
E_{\foreign}(B,\IT)&=(\IT B)^{\frac2\alpha}\pi K \rho(\alpha,\IT B ),\text{ }
E_{\native}(B)=\frac{\pi K}{1+(B\IT)^{-1}}\\
\iftoggle{SC}{\text{Now, \hspace{1.2in}}}{\text{Now,}\nonumber\\} E_{\pu}\left(\frac{\SThres u^{\alpha} }{ \AntGain{\pu}{1} }  \AntGain{\pu}{k} ,u^{{\alpha_{t_0}}}\right)&=
\pi M u^2 \SThres ^{\frac2{\alpha}}\left(\frac{\AntGain{\pu}{k}}{ \AntGain{\pu}{1} }\right)^{\frac2{\alpha}}\rho\left(\alpha,\frac{\SThres \AntGain{\pu}{k}}{ \AntGain{\pu}{1} }  \right),\\
E_{\foreign}\left(\frac{\SThres u^\alpha \AntGain{\su}{k}}{ \AntGain{\pu}{1} } ,\IT\right)&=\pi K \SThres^{\frac2{\alpha}} u^2\IT ^{\frac2{\alpha}} \left( \frac{\AntGain{\su}{k}}{ \AntGain{\pu}{1} } \right)^{\frac2\alpha}
\rho\left(\alpha,\IT \frac{\SThres u^\alpha \AntGain{\su}{k}}{ \AntGain{\pu}{1} } \right),\\
 E_{\native}\left(\frac{\SThres u^\alpha\AntGain{\su}{k}}{ \AntGain{\pu}{1} } \right)&=\frac{\pi K}{1+(\frac{\SThres \AntGain{\su}{k}}{ \AntGain{\pu}{1} } \IT)^{-1}u^{-\alpha}}
\end{align*}}
{\begin{align*}
E_{\pu}(B,e)&=
\pi M B^{\frac2{\alpha}}\rho(\alpha,B/e),\text{ }\\
E_{\foreign}(B,\IT)&=(\IT B)^{\frac2\alpha}\pi K \rho(\alpha,\IT B ),\\
E_{\native}(B)&=\frac{\pi K}{1+(B\IT)^{-1}}
\end{align*}
Now,
\begin{align*}
 &E_{\pu}\left(\frac{\SThres u^{\alpha} }{ \AntGain{\pu}{1} }  \AntGain{\pu}{k} ,u^{{\alpha_{t_0}}}\right)=
\pi M u^2 \SThres ^{\frac2{\alpha}}\left(\frac{\AntGain{\pu}{k}}{ \AntGain{\pu}{1} }\right)^{\frac2{\alpha}}\rho\left(\alpha,\frac{\SThres \AntGain{\pu}{k}}{ \AntGain{\pu}{1} }  \right),\\
&E_{\foreign}\left(\frac{\SThres u^\alpha \AntGain{\su}{k}}{ \AntGain{\pu}{1} } ,\IT\right)=\pi K u^2 \left(\tau\IT \frac{\AntGain{\su}{k}}{ \AntGain{\pu}{1} } \right)^{\frac2\alpha}
\rho\left(\alpha,\IT \frac{\SThres u^\alpha \AntGain{\su}{k}}{ \AntGain{\pu}{1} } \right),\\
 &E_{\native}\left(\frac{\SThres u^\alpha\AntGain{\su}{k}}{ \AntGain{\pu}{1} } \right)=\frac{\pi K}{1+(\frac{\SThres \AntGain{\su}{k}}{ \AntGain{\pu}{1} } \IT)^{-1}u^{-\alpha}}
\end{align*}}
Then, the  coverage probability is given as
\iftoggle{SC}{
\begin{align*}
\Pc^\mathrm{c}_\mathrm{S}(\SThres)&=2\pi\lambdaPT M
\int_0^\infty\expU{-\frac{\SThres {u}^{\alpha} }{{{\IT}}\AntGain{\su}{1} }\NoiseS
-u^2 \pi \SThres ^{\frac2{\alpha}}\left(
		\lambdaPT \sum_{k=1}^2 b_k
		 M  \left(\frac{\AntGain{\pu}{k}}{ \AntGain{\pu}{1} }\right)^{\frac2{\alpha}}\rho\left(\alpha,\frac{\SThres \AntGain{\pu}{k}}{ \AntGain{\pu}{1}}  \right)\right)}\\
		 &\expU{-u^2 \pi \SThres ^{\frac2{\alpha}}\left(
		 \lambdaST \sum_{k=1}^2 a_k
		K  \IT ^{\frac2{\alpha}} \left( \frac{\AntGain{\su}{k}}{ \AntGain{\pu}{1} } \right)^{\frac2\alpha}
\rho\left(\alpha,\IT \frac{\SThres u^\alpha \AntGain{\su}{k}}{ \AntGain{\pu}{1} } \right)
	\right)
	-\pi \lambdaST \sum_{k=1}^2 a_k
	\fracS{K}{\left(1+\left(\frac{\SThres \AntGain{\su}{k}}{ \AntGain{\pu}{1} } \IT\right)^{-1}u^{-\alpha}\right)}	
	-	\pi\lambdaPT Mu^2}
u\intd u.
\end{align*} }
{\begin{align*}
\Pc^\mathrm{c}_\mathrm{S}(\SThres)&=2\pi\lambdaPT M
\int_0^\infty\expU{-\frac{\SThres {u}^{\alpha} }{{{\IT}}\AntGain{\su}{1} }\NoiseS}\\
&\expU{-u^2 \pi \SThres ^{\frac2{\alpha}}\left(
		\lambdaPT \sum_{k=1}^2 b_k
		 M  \left(\frac{\AntGain{\pu}{k}}{ \AntGain{\pu}{1} }\right)^{\frac2{\alpha}}\rho\left(\alpha,\frac{\SThres \AntGain{\pu}{k}}{ \AntGain{\pu}{1}}  \right)\right)}\\
		 &\expU{-u^2 \pi \SThres ^{\frac2{\alpha}}\left(
		 \lambdaST \sum_{k=1}^2 a_k
		K  \IT ^{\frac2{\alpha}} \left( \frac{\AntGain{\su}{k}}{ \AntGain{\pu}{1} } \right)^{\frac2\alpha}
\rho\left(\alpha,\IT \frac{\SThres u^\alpha \AntGain{\su}{k}}{ \AntGain{\pu}{1} } \right)
	\right)}\\
	&\expU{
	-\pi \lambdaST \sum_{k=1}^2 a_k
	\fracS{K}{\left(1+\left(\frac{\SThres \AntGain{\su}{k}}{ \AntGain{\pu}{1} } \IT\right)^{-1}u^{-\alpha}\right)}	
	-	\pi\lambdaPT Mu^2}
u\intd u.
\end{align*}}

Assuming similar assumptions for the \primary and \secondary antennas as taken in the \secondary case, we can get insights about how antenna beamwidth affects the primary performance.


\textbf{Impact of narrowing \primary antenna beamwidth:}
The term denoting \secondary interference decreases with $\nAntP$ as $u^2c^{\frac2\alpha}\rho\left(\alpha,u^\alpha\frac{c}\nAntP\right)\frac{\kappa}{\nAntP^{2/\alpha}}+\frac{1}{1+\nAntP u^{-\alpha}/c}$ ($\approx\frac{1}{\nAntP}\frac{c}{u^\alpha}$ as $\nAntP\rightarrow\infty$). Here, $c$ is some variable independent of $\nAntP$. Similarly, the term denoting the \primary interference decreases with $\nAntP$ as $\frac{\kappa}\nAntP\rho(\alpha,\SThres)+\left(\frac{1-\kappa}{\nAntP}\right)^{2/\alpha}\rho\left(\alpha,\SThres\frac{1-\kappa}{\nAntP}\right)$. Therefore, narrowing the \primary antennas beamwidth improves the \primary performance significantly.

\textbf{Impact of narrowing \secondary antenna beamwidth:} Here, the term denoting the \secondary interference changes with $\nAntS$ as 
$u^2d^{\frac2\alpha}\rho\left(\alpha,u^\alpha\nAntS{d}\right)\frac{\kappa}{\nAntS^{1-2/\alpha}}+\frac{\kappa}{\nAntS}\frac{1}{1+ u^{-\alpha}/(\nAntS d)}
+u^2((1-\kappa)d)^{\frac2\alpha}\rho\left(\alpha,u^\alpha(1-\kappa){d}\right)+\fracS{1}{\left(1+ u^{-\alpha}/((1-\kappa)d)\right)}$, where $d$ is some variable independent of $\nAntS$. The term denoting the \primary interference remains unchanged with with $\nAntS$. Therefore, narrowing the \secondary beamwidth has very little affect on the \primary performance. 
  


\subsection{Rate Coverage for the Primary and Secondary Operators}

In this section, we derive the downlink rate coverage which is defined as the probability that the rate of a typical user is greater than the threshold $\RThres$, 
$\mathrm{R}^\mathrm{c}(\RThres)=\prob{\mathrm{Rate}>\RThres}$. 

Let $O_\su~(\text{or} ~O_\pu)$ denote the time-frequency resources allocated to each user associated with the `tagged' BS of a \secondary user (or a \primary user). The instantaneous rate of the considered typical \secondary user can then be written as $R_{\su}=O_\su\log{\left(1+\SINR_{\su}\right)}$.
 The value of $O_\su$  depends upon the
number of users ($n_\su$), equivalently the load, served by the tagged BS.
The load  $n_\su$ is a random variable due to the randomly sized coverage areas of each BS and random number of users in the coverage areas. As shown in \cite{SinDhiJ2013,SinghBackHaul2015}, approximating this load with its respective mean does not compromise the accuracy of results. 
Since the user distribution of each network is assumed to be PPP, the average number of users associated with the tagged  BS of each networks associated with the typical user can be modeled similarly to \cite{SinDhiJ2013,SinghBackHaul2015}:
$n_\su=1+1.28\frac{\lambdaSR}{\lambdaST}$ and $n_\pu=1+1.28\frac{\lambdaPR}{\lambdaPT}$.
Now, we assume that the scheduler at the tagged BS gives  $1/n$ fraction of resources to each user. This assumption can be justified as most schedulers such as round robin or proportional fair  give  approximately  $1/n_\su$ (or $1/n_\pu$) fraction of resources to each user on average.
Using the mean load approximation, the instantaneous rate of a typical \secondary user which is associated with BS at $\y_0$ is given as
\begin{align}
R_\su&=\frac{W}{n_\su}\log{\left(1+\SINR_\su\right)}\label{eq:InRateExp}.
\end{align} 
Now, $\Rc^\mathrm{c}_\su(\RThres)$ and $\Rc^\mathrm{c}_\pu(\RThres)$ can be derived in terms of coverage probability as follows:
\begin{align}
\Rc^\mathrm{c}_\su(\RThres)&=\prob{R_\su>\RThres}=\prob{\frac{W}{n_\su}\log{(1+\SINR_\su)}>\RThres}\nonumber\\
&=\prob{\SINR_\su>2^{\RThres\frac{n_\su}{W}}-1}=\Pc_\su^\mathrm{c}\left(2^{\RThres n_\su/W}-1\right),\nonumber\\
\Rc^\mathrm{c}_\pu(\RThres)&=\Pc_\pu^\mathrm{c}\left(2^{\RThres n_\pu/W}-1\right)\label{eq:Rc}.
\end{align}

\newcommand{\ase}{median rate }
\newcommand{\SumR}[1]{\mathcal{R}_{#1}}

 We now define the \ase which works as a proxy to the network performance. 

\begin{definition} Let $\mathcal{B}$ denote a region with unit area. The \ase $\SumR{}$ of an operator is define as the sum of the median rates of all the users served in $\mathcal{B}$, which is
\begin{align}
\SumR{}=\expect{\sum_{u\in\PPR\cap\mathcal{B}}\mathbb{M}_u\left[R\right]}
\end{align}
where $\mathbb{M}_u\left[R\right]$ is the median rate of the user at $u$.
 \end{definition}

 From the stationarity of the user PPP, 
 \begin{align}
\SumR{}=\expect{\sum_{u\in\PPR\cap\mathcal{B}}\mathbb{M}_u\left[R\right]}=\lambdaR\int_{\mathcal{B}}\mathbb{M}^0\left[R\right]\intd u=\lambdaR\mathbb{M}^0\left[R\right]
\end{align}
where $\mathbb{M}^0$ denotes the median rate at the origin  under Palm ({\em i.e.} conditioned on the fact that there is a user at 0). Note that this is equal to the rate threshold where rate coverage of the typical user at the origin is 0.5.
Let ${(\Pc^\mathrm{c})}^{-1}(\cdot)$ denote the inverse of ${\Pc^\mathrm{c}}(\cdot)$. Now using \eqref{eq:Rc}, we can compute the \ase of the \primary and \secondary operators as follows:
\begin{align}
\SumR{\pu}&=W\frac{\lambdaPR}{1+1.28\fracS{\lambdaPR}{\lambdaPT}}\log\left(1+{(\Pc_\pu^\mathrm{c})}^{-1}\left(0.5\right)\right)\\
\SumR{\su}&=W\frac{\lambdaSR}{1+1.28\fracS{\lambdaSR}{\lambdaST}}\log\left(1+{(\Pc_\su^\mathrm{c})}^{-1}\left(0.5\right)\right).
\end{align}

\section{License Pricing and Revenue Model} \label{sec:Licensing}
In this section, we present the utility model, and describe the general license pricing and revenue functions. We assume a centralized  licensing model in which a central entity, such as FCC, has a control over the licensing for the \primary and \secondary operators. Therefore, even though the primary operator has an "exclusive-use" license, the decision to sell a restricted license to a secondary operator is taken by both the primary operator and the central licensing authority. These two entities will also share the revenue of the restricted secondary license. 

Let  $\PRevFun(\SumR{\pu})$ define the per-unit-area revenue function of the \primary operator from its own users when it provides a sum rate of $\SumR{\pu}$.  Similarly, we define the \secondary revenue  function  $\SRevFun(\cdot)$  that models the revenue of the \secondary network from its users. One special case is the linear mean revenue function, which is given  as follows
\noindent\begin{align}
\PRevFun(\SumR{\pu})&=\PRevConst\SumR{\pu},&
\SRevFun(\SumR{\su})&=\SRevConst\SumR{\su},
\end{align}
with $\PRevConst$ and $\SRevConst$ representing the linear \primary and \secondary revenue constants.

To characterize the licensing cost, we assume the licenses are given on a unit area region basis. Let the \primary licensing function $\PLicFun(\SumR{\pu})$  denote the license price paid by the \primary to central entity when it provides the \ase of  $\SumR{\pu}$ to its users. Similarly, we define \secondary licensing  function  $\SLicFun(\cdot)$ which denotes the price paid by the \secondary operator to the central entity. We also assume that \secondary operator has to pay some license price to the \primary operator as an incentive to let it use the \primary license band which is given as $\SPLicFun(\SumR{\su})$. We also define a special case as linear licensing function where the licensing cost paid by the \primary and the \secondary operators to the central entity and by the \secondary operator to the \primary operator are given as
\begin{align}
\PLicFun(\SumR{\pu})&=\PLicConst\SumR{\pu},&
\SLicFun(\SumR{\su})&=\SLicConst\SumR{\su},\iftoggle{SC}{&}{\nonumber\\}
\SPLicFun(\SumR{\su})&=\SPLicConst\SumR{\su}.\nonumber
\end{align}

The utility function of an entity is defined by its total revenue which for the three entities is given as follows:
\begin{align} 
\PUtFun(\SumR\pu)&=\PRevFun(\SumR\pu)-\PLicFun(\SumR\pu)+\SPLicFun(\SumR\su), \label{eq:utilities_p}\\
\SUtFun(\SumR\su)&=\SRevFun(\SumR\su)-\SLicFun(\SumR\su)-\SPLicFun(\SumR\su), \label{eq:utilities_s}\\
\CUtFun(\SumR\su)&=\PLicFun(\SumR\pu)+\SLicFun(\SumR\su). \label{eq:utilities_c}
\end{align}

Note that the \secondary \ase depends on the maximum interference limit $\IT$. By increasing this limit, \secondary network can increase its \ase for which it has to pay more to central entity and the \primary operator. Increasing this limit, however, decreases the \primary \ase which impacts the \primary network revenue from its own users. Therefore, there exists a trade-off when varying the interference limit $\IT$.

\section{Simulation Results and Discussion} \label{sec:Results}

In this section, we provide numerical results computed from the analytical expressions derived in previous sections, and draw insights into the performance of restricted \secondary licensing in mmWave systems. For these numerical results, 
we adopt an exponential blockage model, i.e., the LOS link probability is determined by $p_\L(x)=\exp(-x/\beta)$, with a LOS region $\beta=150$m. The LOS and NLOS pathloss exponents are $\alpha_\mathrm{L}=2.5$ and $\alpha_\mathrm{N}=3.5$, and the corresponding gains are $C_\mathrm{L}=C_\mathrm{N}=-60$dB. Unless otherwise mentioned, the \primary network has an average cell radius of $100$m, which is equivalent to a BS density of $\approx 30/\text{km}^2$. The transmit power of the \primary BSs is $40$dBm, while the transmit power of each \secondary BS is determined according to \eqref{eq:Sec_Power} to ensure that its average interference on its home \primary user in less than the threshold $\IT$. Both networks operate at $28$GHz carrier frequency over a shared bandwidth of $500$MHz. Note that the noise power at the BS is $-110$dB. Therefore, if  $\IT$ is between $-110$dB and $-120$dB, the secondary interference will be in the order of the noise. For the antenna patterns, the \primary and \secondary BSs employ a sectored beam pattern models as described in Section \ref{subsec:Channel_Model}. First, we verify the the derived analytical results for the \primary and \secondary coverage probabilities, before delving into the spectrum sharing rate and utility characterization.

\subsection{Coverage and Rate Results} \label{subsec:sim_cov} 
\begin{figure}[t!]
	\centering
	\includegraphics[width=\scf\columnwidth,trim=00 0 0 0,clip=true]{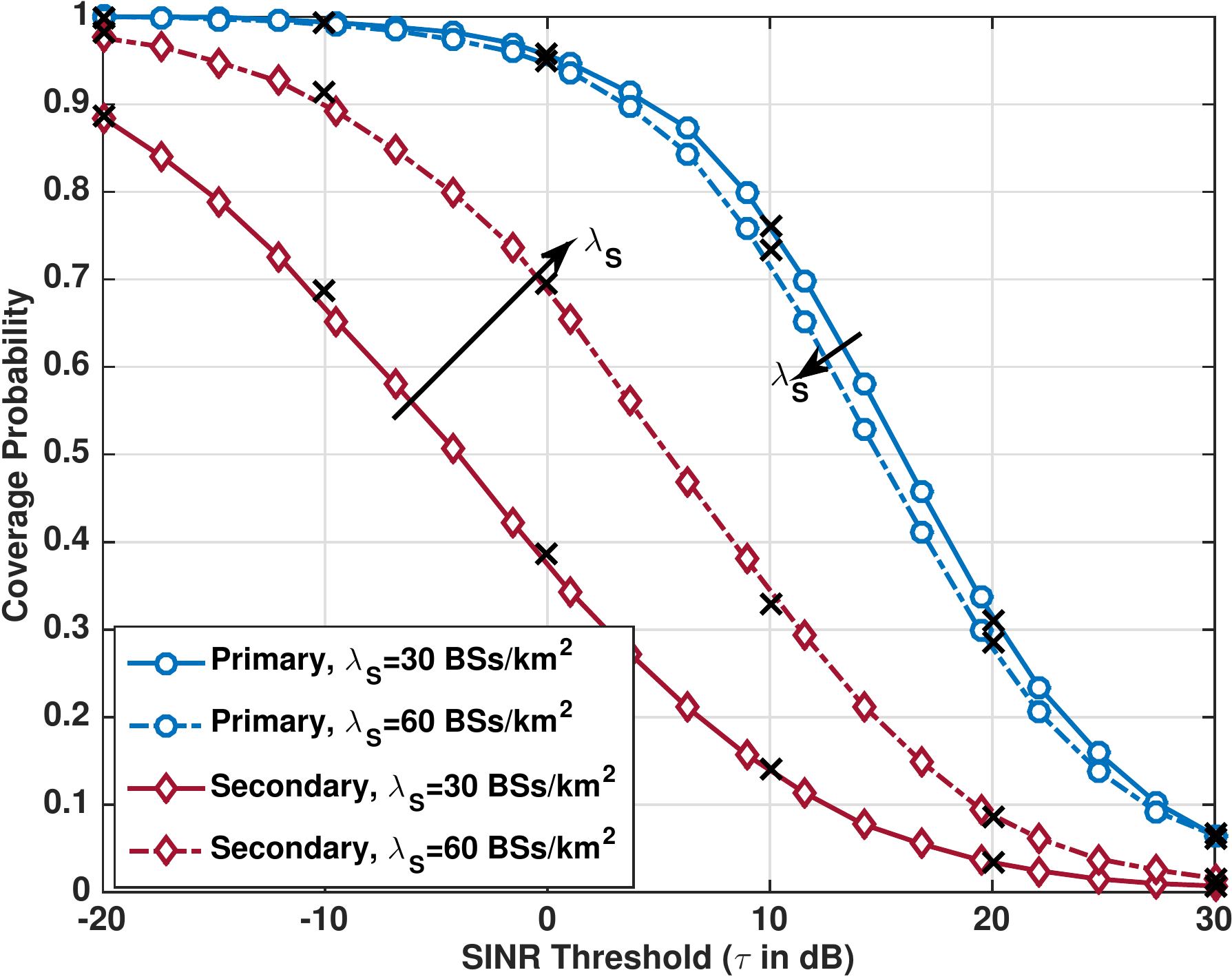}
	\caption{Coverage probabilities of the  two licensees for two different values of $\lambda_\su$ with $\IT=-120$dB. The  secondary can  improve its coverage by choosing an appropriate density without impacting the primary coverage.}
	\label{fig:cov_impactdensity}
\end{figure}

Since the \secondary operator shares the same time-frequency resources with the primary, it is important to characterize the impact of sharing on the primary performance. In this subsection, we evaluate the coverage and rate for both operators, and study the impact of secondary network's densification and narrowing the beamforming beams on the performance of the two networks. 


\textbf{Validation of analysis and impact of the \secondary densification:}  Fig. \ref{fig:cov_impactdensity} shows the coverage probabilities of  both operators for two different values of the secondary density, $\lambdaST=30$ BSs/km$^2$ and $\lambdaST=60$ BSs/km$^2$. The density of the primary BSs is fixed at $\lambdaPT=30$ BSs/km$^2$ and the maximum \secondary interference threshold is set to -120 dB. We can see that despite the various assumptions taken in the analysis, the analysis matches the simulations closely. An interesting note from Fig. \ref{fig:cov_impactdensity} is that increasing the secondary network density significantly improves the secondary network coverage while causing a negligible impact on the primary network performance. In particular, when $\lambdaST$ increases from 30 to 60 BSs/km$^2$, the median SINR of the secondary network increases from -4dB to 6dB while the median SINR of primary network decreases only by 2 dB. This indicates that in mmWave, both primary and secondary can achieve significant coverage probability by selecting appropriate values of $\IT$ and BS densities.


\begin{figure}[t!]  
	\centering
	\includegraphics[width=\scf\columnwidth,trim=0 0 0 0,clip]{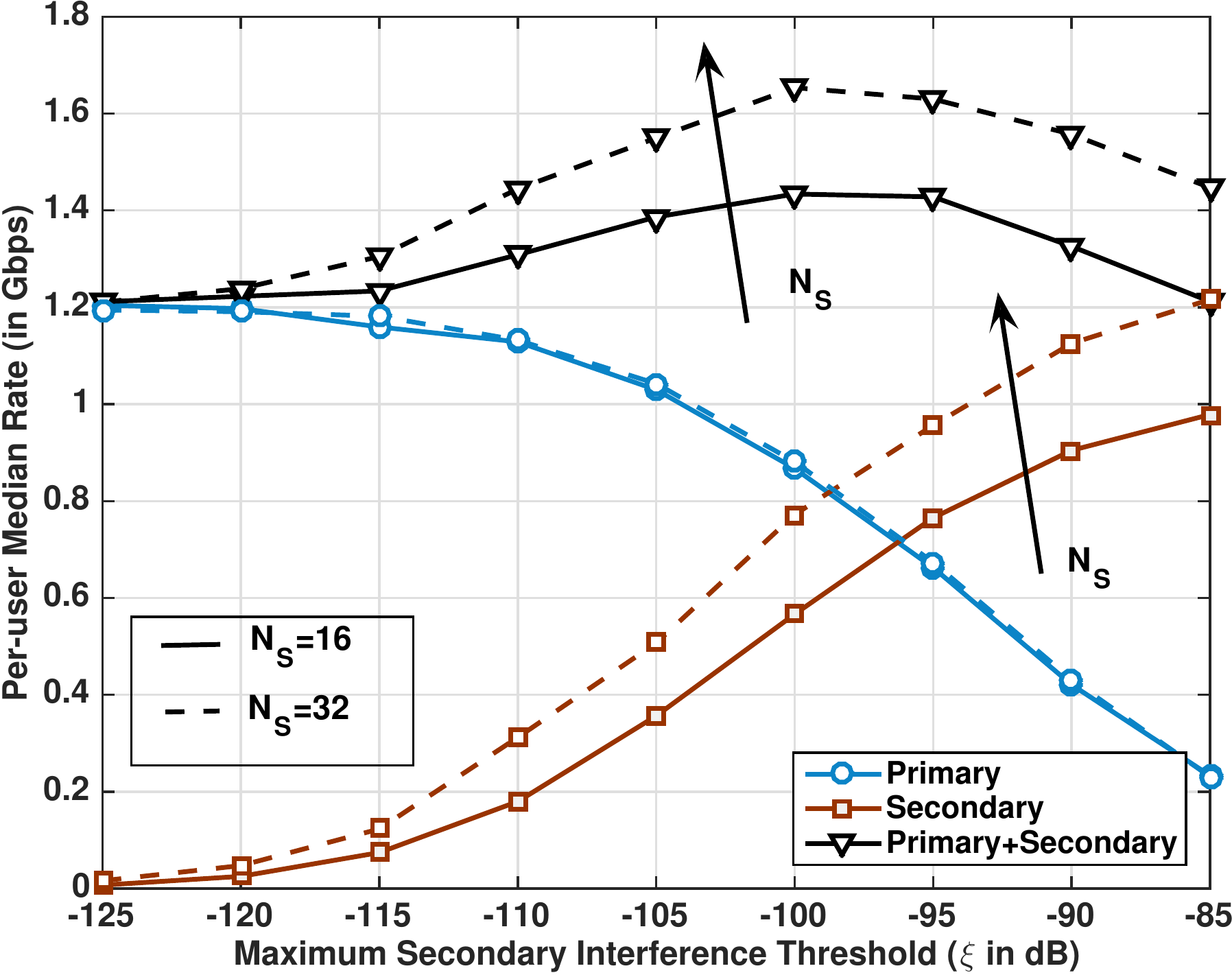}
	\caption{Sum median rate of the primary and secondary networks as well as the total sum median rate versus $\IT$. The  networks have equal density of 56 BS/km$^2$, which corresponds to an average cell radius of $75$m.}
	\label{fig:asevsIT}
\end{figure}

\textbf{Impact of the \secondary antenna beamwidth:} One important feature of mmWave systems is their ability to use large antennas arrays and narrow directional beams. To examine the impact of antenna beamwidth, we plot the median per-user rate of both the primary and secondary networks along with their sum-rate for two different values of number of secondary antennas in Fig. \ref{fig:asevsIT}. These rates are plotted versus the secondary interference threshold $\IT$. First, Fig. \ref{fig:asevsIT} shows that the secondary network performance improves as the number of its BS antennas increase (or equivalently as narrower beams are employed). Another interesting note is that the primary performance is almost invariant of the secondary antennas beamwidth. This means that the secondary network can always improve its performance by employing narrower beamforming beams without impacting the primary performance. This will also lead to an improvement in the overall system performance. Finally, we note that for every secondary BS beamwidth, there exists a finite value for the interference threshold $\IT$ at which the sum-rate is maximized. Therefore, this threshold need to be wisely adjusted for the spectrum sharing network based on the different network parameters to guarantee achieving the best performance.


\begin{figure}[t!]
	\centering
	\includegraphics[width=\scft\columnwidth]{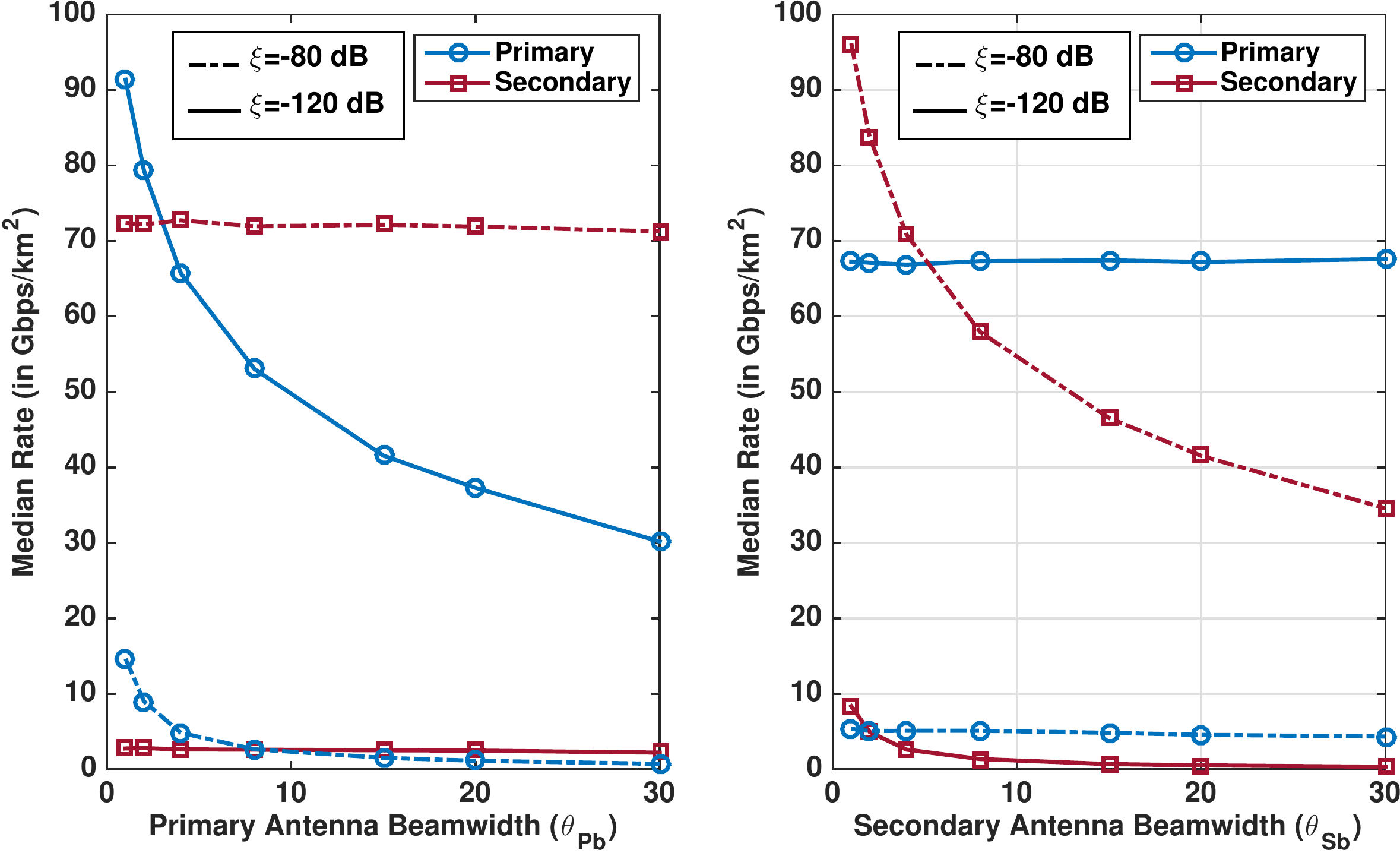}
	\caption{Effect of \primary and \secondary antenna beamwidth over primary and secondary operators for two values of interference threshold $\IT$. Both operators have equal density of 60 BSs/km$^2$. Secondary antenna beamwidth  significantly improves its own performance but does not impact primary performance. Therefore, secondary antennas can be made narrow to get high rates without causing additional interference on primary. Similar trends can also be observed for primary antennas. }
	\label{fig:Rate_BW}
\end{figure}

To verify the insights drawn from the analytical expressions about narrowing the \primary and \secondary beamforming beamwidth in Sections \ref{subsec:SecondaryCov} - \ref{subsec:PrimaryCov}, we plot the \primary and \secondary median rates versus the BS antenna beamwidth in Fig. \ref{fig:Rate_BW}. This figure shows the narrowing the beams of the BSs in one network (\primary or secondary) improves the performance of this network with almost no impact on the other network performance. This trend happens even with higher secondary interference threshold as depicted in Fig. \ref{fig:Rate_BW}.



\textbf{Comparison with uncoordinated spectrum sharing:}
Now, we compare the gain from restricted secondary licensing proposed in this paper over the uncoordinated spectrum sharing considered in \cite{GuptaAndHeath2016}. We consider a scenario where two operators buy exclusive licenses to two different mmWave bands with equal bandwidth. The two operators decide to share their licenses in the following way: each operator is known as a primary in its own band and a secondary in the other operator’s band. In the restricted secondary licensing, each operator can transmit in other operator bands with the restriction on its transmit power. In the uncoordinated sharing, the two operators are allowed to  transmit in each other bands with no restriction. For simplicity, we assume that the two operators, in the uncoordinated sharing case, have the same transmit power. To have a fair comparison, we choose the transmit power in uncoordinated case such that the total power (sum of the transmit power of the two operators) is equal to the total power of the restricted secondary sharing case. Fig. \ref{fig:PSLicVsULic100} compares the median rates of  an operator achieved in its primary and secondary bands  as well as its aggregate median rate for the two sharing cases. 
First, this figure shows that restricted secondary licensing can achieve higher sum rates compared to uncoordinated sharing if the interference threshold is appropriately adjusted. The figure also indicates that the restricted licensing approach provides a mean for differentiating the access to guarantee that the primary user gets better performance in its band. This is captured by the higher rate of the primary operator in the restricted secondary licensing case compared to the primary rate in the uncoordinated sharing for wide range of $\IT$ values. 

In Fig. \ref{fig:PSLicVsULic115}, we show the impact of secondary network density ($\lambdaST$) on the gain of restricted secondary licensing over uncoordinated sharing. Fig.  \ref{fig:PSLicVsULic115} illustrates that increasing $\lambdaST$
 decreases the rate of the primary operator in two sharing approaches, which is expected. Interestingly, the degradation in the primary performance is smaller in the restricted licensing case which leads to higher overall gain compared to the uncoordinated sharing. This also means that the gain of restricted licensing over uncoordinated sharing increases in dense networks, which is particularly important for mmWave systems. In conclusion, the results in Fig. \ref{fig:PSLicVsULic100} - Fig. \ref{fig:PSLicVsULic115} indicate that static coordination is in fact beneficial for mmWave dense networks as it leads to higher rates and provides a way of differentiating the access between the spectrum sharing operators.

\iftoggle{SC}{
\begin{figure}[t!]
	\centering
	\subfigure[center][{}]{
		{\includegraphics[width=0.475\columnwidth,trim=0 0 0 0,clip=true]{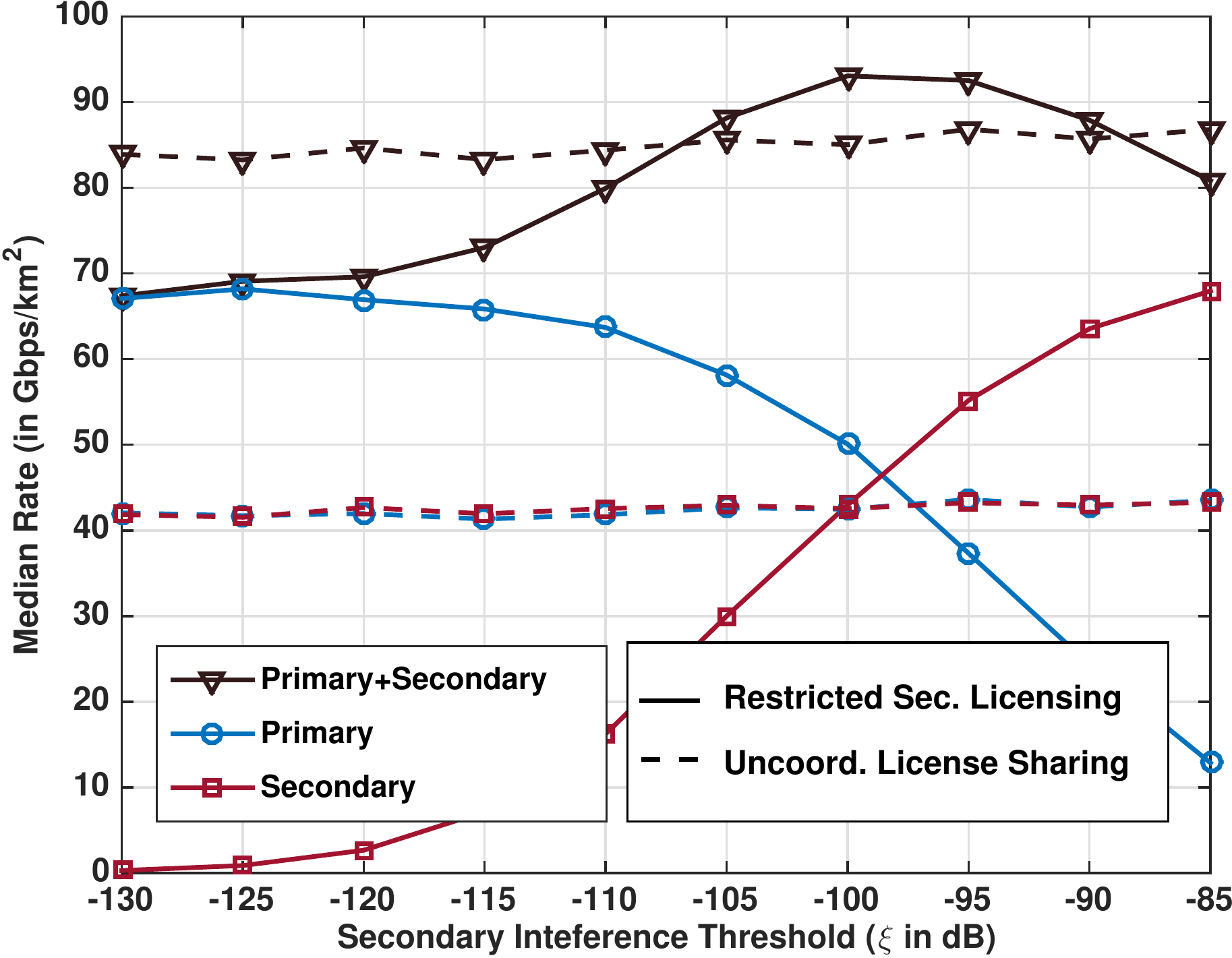}}

	\label{fig:PSLicVsULic100}}
	\subfigure[center][{}]{
		{\includegraphics[width=0.47\columnwidth,trim=0 8 0 0]{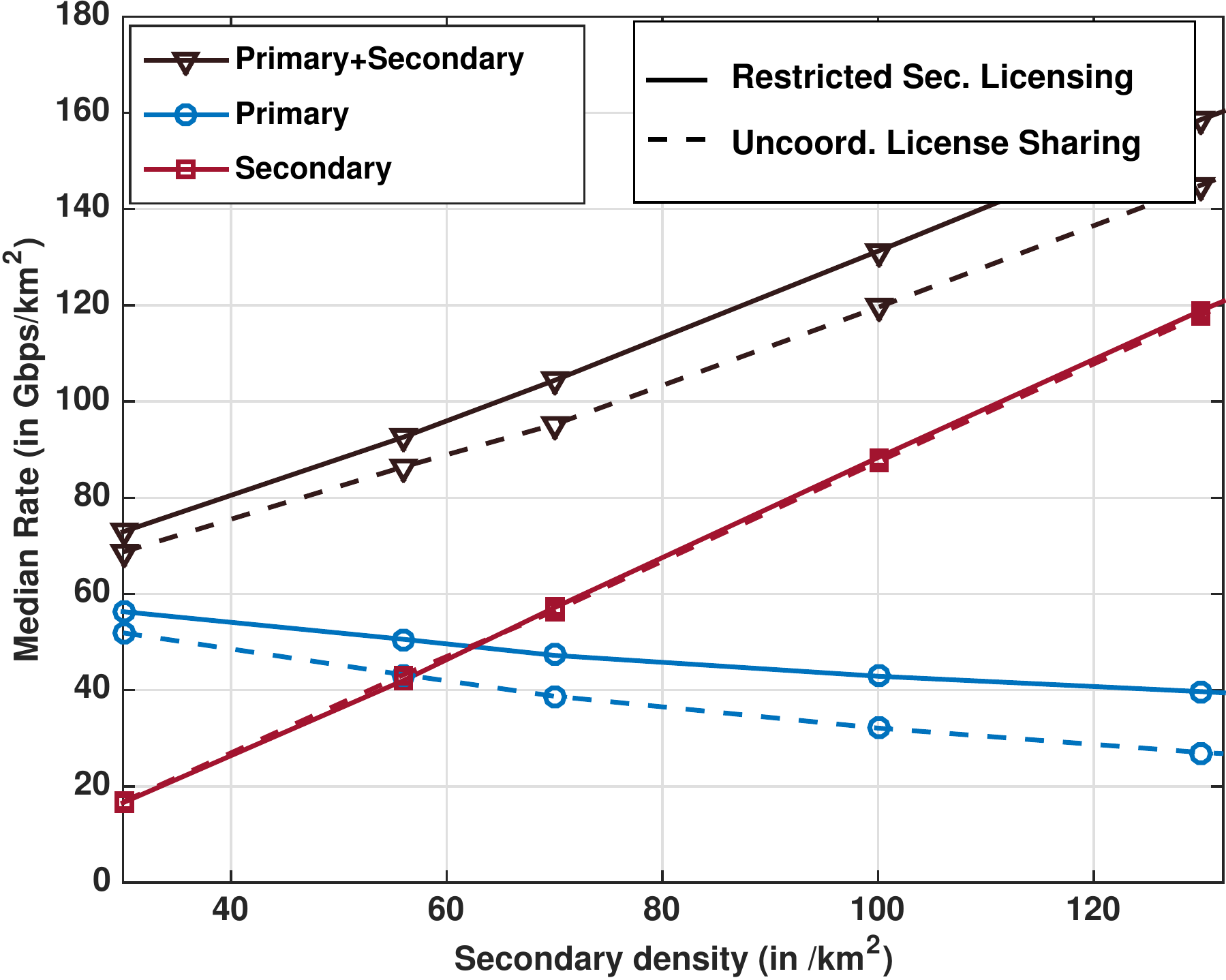}}
		\label{fig:PSLicVsULic115}}
	\caption{Comparison of restricted secondary licensing over uncoordinated sharing. (a) Variation of median rates of an operator in the primary and secondary bands and its sum median rate with $\IT$. Both operators have equal density of 60 BSs/km$^2$.  Secondary licensing can achieve higher sum rates compared to uncoordinated sharing if $\IT$ is appropriately adjusted. (b) Variation of median rates of an operator in the primary and secondary bands and its sum median rate with secondary density $\lambdaST$. Primary BS density is kept constant  at  60 BSs/km$^2$. The gain of restricted secondary licensing over uncoordinated sharing increases with $\lambdaST$.}
	\label{fig:Fig1x}
\end{figure} 
}
{\begin{figure}[t!]
	\centering
	\subfigure[center][{}]{
		{\includegraphics[width=0.9\columnwidth,trim=0 0 0 0,clip=true]{figs/Fig10PSLicVsULic.eps}}	
	\label{fig:PSLicVsULic100}}
	\subfigure[center][{}]{
		{\includegraphics[width=0.9\columnwidth,trim=0 8 0 0]{figs/Fig115PSLicVsULic.eps}}
		\label{fig:PSLicVsULic115}}
	\caption{Comparison of secondary licensing over uncoordinated sharing. (a) Variation of median rates of an operator in the primary and secondary bands and its sum median rate with $\IT$. Both operators have equal density of 60 BSs/km$^2$.  Restricted secondary licensing can achieve higher sum rates compared to uncoordinated sharing if the interference threshold is appropriately adjusted. (b) Variation of median rates of an operator in the primary and secondary bands and its sum median rate with secondary BS density. Primary BS density is kept constant  at  60 BSs/km$^2$. The gain of restricted secondary licensing over uncoordinated sharing increases with secondary density.}
	\label{fig:Fig1x}
\end{figure} 
}

\subsection{Primary and Secondary Utilities: The Benefits of Spectrum Sharing} \label{subsec:sim_utility}
In this subsection, we explore the potential gains of \secondary licensing in mmWave cellular systems. 
We adopt the pricing model from Section \ref{sec:Licensing}, with revenue constants $\PRevConst=1, \SRevConst=1$, and licensing cost constants $\PLicConst=0.25, \SLicConst=0.125$.

\begin{figure}[t!]
	\centering
	\includegraphics[width=\scf\columnwidth]{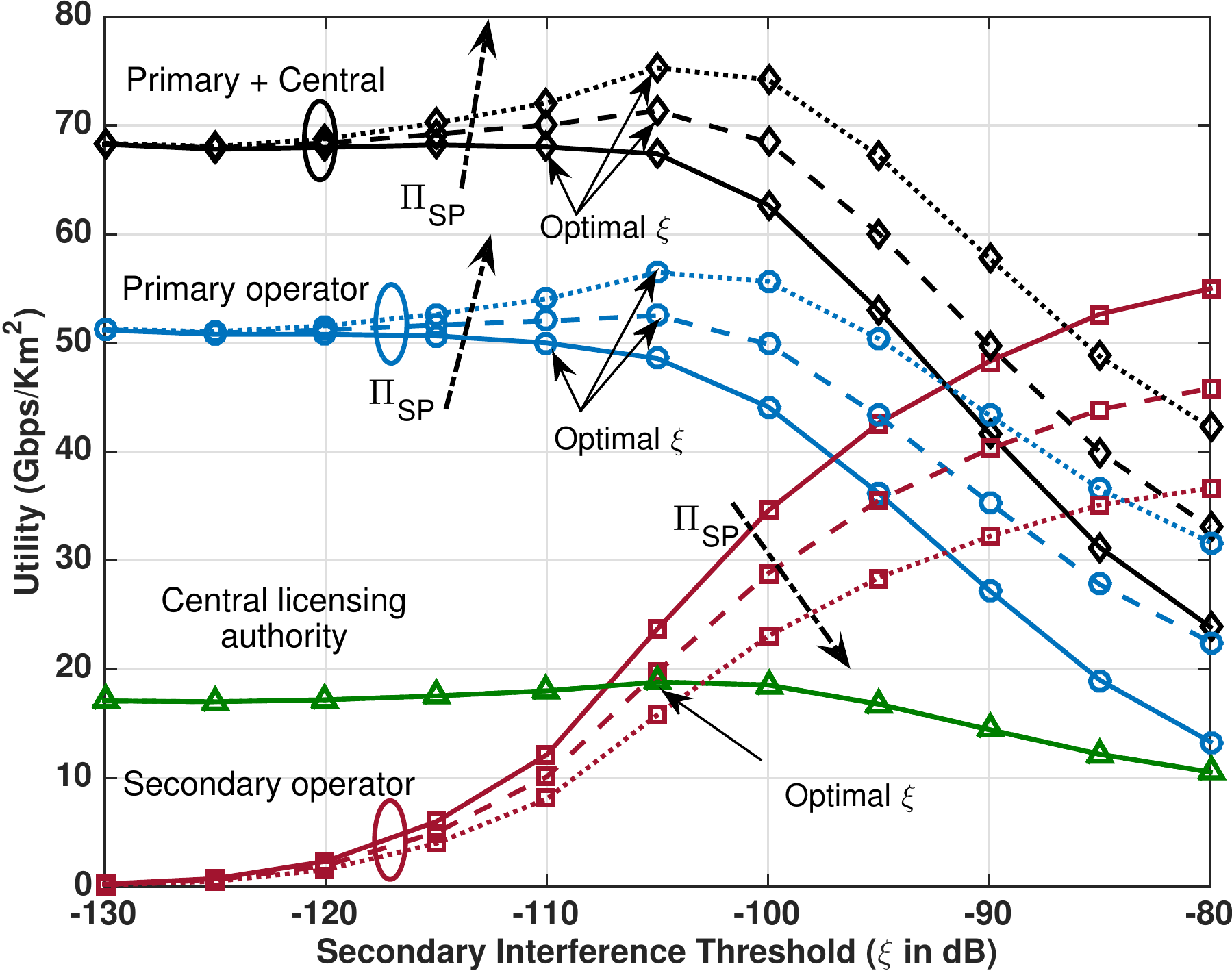}
	\caption{Utility of \primary network, the central licensing authority and the sum of the former two utilities as function of the maximum \secondary interference threshold $\IT$. The optimal threshold for the total utilities falls in between the optimal thresholds of utilities of the \primary and the central entity $\Pi_{SP}=.125,.25,.375$.}
	\label{fig:tot_utilities}
\end{figure}

\textbf{Gain of the primary network from restricted \secondary licensing:} In Fig. \ref{fig:tot_utilities}, we plot the utility functions of the \primary operator, the \secondary operator, and the central licensing authority, defined in \eqref{eq:utilities_p}-\eqref{eq:utilities_c}, versus the \secondary interference threshold  $\IT$ for three different values of the secondary-to-\primary licensing constant $\SPLicConst$. In this result, we consider a \primary network of density  $30/\text{km}^2$, and a \secondary network of density of $60/\text{km}^2$. First, the figure shows that increasing $\IT$ improves the \secondary operator utility which is expected. Interestingly, the utility of the \primary network does not always decrease with increase in $\IT$. The figure indicates that $\IT$ that maximizes the \primary network utility is finite, which means that the \primary network can actually benefit from the restricted \secondary licensing. The intuition is that the \secondary network needs to pay for its interference to the \primary network. As this interference increases, the money that the \primary  network gets from the restricted secondary licensing is more that its revenue from its own network. This underling trade-off normally yields an optimal value for the \secondary interference threshold that maximizes the \primary network utility. This means that the primary network has clear incentive to share its spectrum using restricted secondary licensing. 

\textbf{Joint optimization of the \primary and the central entity:} The utility of the central licensing authority remains constant for different values of $\SPLicConst$, which can be noted from \eqref{eq:utilities_p}-\eqref{eq:utilities_c}. As the value of $\IT$ that maximizes the utility of the central authority can be larger than that maximizing the \primary utility as shown in Fig. \ref{fig:tot_utilities}, the central licensing authority has the incentive to push the primary to share with more degradation than the primary would otherwise share. Fig. \ref{fig:tot_utilities}, also plots the total utility function which defined as the sum of the \primary and central licensing authority's utilities. Intuitively, the optimal threshold for the total utility falls in between the optimal thresholds of the \primary and central entity utilities. 




\section{Conclusion} \label{sec:Conc}
In this paper, we modeled a mmWave cellular system with a \primary operator that has an ``exclusive-use'' license with a provision to sell a restricted secondary license to another operator that has a maximum allowable interference threshold. This licensing approach provides a way of differentiating the spectrum access for the different operators, and hence is more practical. Due to this restriction on the secondary interference, though, the transmit power of a secondary BSs is a random variable. This required developing new analytical tools to analyze the network coverage and rate. Results showed that \secondary can achieve good rate coverage with a small impact on the \primary performance. Results also indicated that narrow beams and dense networks can further improve secondary network performance. Compared to uncoordinated sharing, we showed that a reasonable gain can be achieved with the proposed static coordinated sharing approach. Further, restricted secondary licensing can guarantee a certain spectrum access quality for the primary user, which is not the case in uncoordinated sharing. We also considered a revenue model for both operators in the presence of a central licensing authority. Using this model, we showed that the \primary operator can achieve good benefits from restricted secondary licensing, and hence has a good incentive to share its spectrum. Results also illustrated that the central licensing authority can get more gain with restricted secondary licensing. As the optimal interference thresholds for the central licensing and primary operators can be different, the central authority may push the primary operator to share with more degradation than the primary would otherwise share. Overall, the primary and secondary operators as well as the central licensing authority can benefit from restricted secondary licensing. For future work, it would be of interest to investigate how techniques like multi-user multiplexing affect the insights on restricted \secondary licensing. It is also important to explore how temporal variations in the traffic demands for the two operators impact the network performance.

\appendices
\section{Derivation of probability distribution of $R_i$'s}\label{app:RiDist}
Here, we compute the joint distribution of $R_i$ and $T_i=\L$. The proof for $T_i=\N$ is similar. Consider the $i^{th}$ \secondary BS. Now, the \primary user PPP can be divided into two independent  PPPs: $\Phi_{1\L}^\rx$ consisting of all  \primary user having LOS link to the $i^{th}$ \secondary BS and $\Phi_{1\L}^\rx$ with all  \primary user having NLOS link to the $i^{th}$ \secondary BS.  
 Now, let $R_{\L i}$ denote the distance of the closest \primary user in  ${\PPPR}_\L$ whose  distribution can computed as follows:
\begs\begin{align}
\prob{R_{\L i}>r}&=\exp\left(-\int_0^\infty \lambdaPR p_\L(r)2\pi r dr\right)=\expU{-\lambdaPR V_\L(r)}\nonumber\\
f_{R_{\L i}}(r)&=\frac{\intd}{\intd r} \prob{R_{\L i}>r}=2\pi \lambdaPR p_\L(r)r\exp(-\lambdaPR V_\L(r))\nonumber
\end{align}\ens
where the first step is from the void probability of the non-homogenous PPP  ${\PPPR}_{\L}$. Similarly the distance distribution of the closest \primary user in ${\PPPR}_{\N}$ can also be computed.
Now, the joint probability of the event $R_i>r$ and the event that 
 $\mathcal{H}_i$
 is a LOS BS ({\em i.e.} $\LinkHome{i}=\L$) is computed as
\begs\begin{align}
&\prob{R_i>r,\LinkHome{i}=\L}=\int_r^\infty f_{R_{\L i}}(u)\prob{C_\N R_{\N i}^{-\alpha_\N}<C_\L u^{-\alpha_\L}}\intd u \DCPbreak \nonumber
=\int_r^\infty f_{R_{\L i}}(u)\prob{R_{\N i}> \left(\frac{C_\N}{C_\L}\right)^{\frac1{\alpha_\N}} u^{\frac{\alpha_\L}{\alpha_N}}}\intd u
\\&=\int_r^\infty f_{R_{\L i}}(u)\prob{R_{\N i}> \exclusionfun{\L}{\N}(u)}\intd u
\DCPbreak=\int_r^\infty 2\pi \lambdaPR p_\L(u)r\exp(-\lambdaPR V_\L(u)-\lambdaPR V_\N(\exclusionfun{\L}{\N}(u)))\intd u\nonumber
\end{align}\ens
where the last step is from the void probability of ${\PPPR}_\N$.
Therefore, the joint distribution can be computed as follows:
\begs\begin{align}
&f_{R_i}(r,\LinkHome{i}=\L)=\frac{\intd}{\intd r} \prob{R_i>r,\LinkHome{i}=\L}\DCPbreak=2\pi \lambdaPR p_\L(r)r\exp(-\lambdaPR V_\L(r)-\lambdaPR V_\N(\exclusionfun{\L}{\N}(r))).\nonumber
\end{align}\ens

\section{Proof of Lemma 1}\label{app:SSINRCoverageStep1}
Let $\Phi$ be an arbitrary PPP. Now let us assign to each $i^{th}$ \secondary BS, a mark $e(\y_i,\Phi)$ as indicator of $\y_i$ being selected as serving BS from  $\Phi$ and another mark $S(\y_i,\Phi)$ as SINR at $\mathrm{SU}_0$ if BS at $\y_i$ is selected for serving and interferers are from $\Phi$, 
\begs\begin{align}
e(\y_i,\Phi)&=\ind{}\left(\frac{\powerS{i}C_{\LinkSS{i}}}{\|\y_i\|^{\alpha_{\LinkSS{i}}}}>\frac{\powerS{j}C_{\LinkSS{j}}}{\|\y_j\|^{\alpha_{\LinkSS{j}}}} \forall j, \in \Phi\right)
\iftoggle{SC}{,&}{\nonumber\\}
S(\y_i,\Phi)&=\frac{\AntGain{\su}{1} \ChannelSS{i} \powerS{i} C_{\LinkSS{i}}\|\y\|^{-\alpha_{\LinkSS{i}}}}{I_\pu+I_\su(\Phi)+\NoiseS}
\end{align}\ens
\noindent
\begin{flalign*}
&\text{where } I_\pu=\sum\limits_{\x_j\in\PPPT} \powerP{}\AntGain{\pu}{}(\theta_j)\ChannelPS{j}C_{\LinkPS{j}} x_j^{-\alpha_{\LinkPS{j}}}, \text{ and } 
\DCPbreak I_\su(\Phi)=\sum\limits_{\y_j\in\Phi} \powerS{j}\AntGain{\pu}{}(\omega_j)\ChannelSS{j}C_{\LinkSS{j}} y_j^{-\alpha_{\LinkSS{j}}}.\hspace{0.7in}
\end{flalign*}
Using the above two indicators, the coverage probability of \UES~can be written as
\begs\begin{align}
\Pc^\mathrm{c}_\mathrm{S}(\SThres)=&\sum\nolimits_{t_0\in\{\L,\N\}}\expects{}{\sum\nolimits_{\y_i\in \Phi_{2t_0}}\ind{}(S(\y_i,\PPST\setminus \y_i)>\SThres,
\DCPbreakI \vphantom{\sum\nolimits_{\y_i\in \Phi_{2t_0}}} e(\y_i,\PPST\setminus\y_i)=1)}\label{eq:Pcalternate}.
\end{align}\ens
This is due to the fact that $e(\y_i,\PPST\setminus\y_i)$ can be 1 only for one BS that is at $\y_0$, therefore  \eqref{eq:Pcalternate} will give the  coverage probability provided by BS at $\y_0$. \eqref{eq:Pcalternate} can be further written as $\iftoggle{SC}{}{\Pc^\mathrm{c}_\mathrm{S}(\SThres)}$
\begs\begin{align}
\iftoggle{SC}{\Pc^\mathrm{c}_\mathrm{S}(\SThres)}{}
&\stackrel{(a)}{=}\sum\nolimits_{t_0\in\{\L,\N\}}\int_0^\infty\lambdaST p_{t_0}(\|\y\|)\iftoggle{SC}{}{\times}\DCPbreak\iftoggle{SC}{}{\hspace{.3in}}\mathbb{P}^{\y!}\left[S(\y,\PPST)>\SThres,e(\y,\PPST)=1\right]\intd \y\nonumber\\
&\stackrel{(b)}{=}\sum\nolimits_{t_0\in\{\L,\N\}}\int_0^\infty\lambdaST p_{t_0}(\|\y\|)\iftoggle{SC}{}{\times}\DCPbreak\mathbb{P}\left[S(\y,\Phi_2)>\SThres \bigconditioned e(\y,\PPST)=1\right]\prob{e(\y,\PPST)=1}\intd \y\label{eq:app1e1}
\end{align}\ens
where $(a)$ is due to the Campbell Mecke theorem and $(b)$ is due to the Slivnyak theorem. Now $\prob{e(\y,\PPST)=1}$ can be computed as
\begs\begin{align*}
\iftoggle{SC}{}{&}\prob{e(\y,\PPST)=1}\iftoggle{SC}{&}{}=\prob{\ind{}\left(\frac{\powerS{0}C_{t_0}}{y^{\alpha_{t_0}}}> \frac{\powerS{j}C_{t_j}}{y_j^{\alpha_{t_j}}}\  \forall \y_j\in\PPST\right)}
\DCPbreak
=\prob{\prod_{\y_j\in\PPST}\ind{}\left(\frac{\powerS{0}C_{t_0}}{y^{\alpha_{t_0}}}> \frac{\powerS{j}C_{t_j}}{y_j^{\alpha_{t_j}}}\right)}\\
&\stackrel{(a)}{=}\prod\nolimits_{t\in\{\L,\N\}}\prob{\prod\nolimits_{\y_j\in{\PPST}_t}\ind{}\left(\fracS{\powerS{0}C_{t_0}}{y^{\alpha_{t_0}}}> \fracS{\powerS{j}C_{t}}{y_j^{\alpha_{t}}} \right)}\\
&\stackrel{(b)}{=}\prod\nolimits_{t\in\{\L,\N\}}\exp\left(
-2\pi\lambdaST\int_0^\infty\expects{\powerS{}}
{\ind{}\left(\fracS{\powerS{}C_t}{u^{\alpha_t}}
\DCPbreakIII\DCspace
>\fracS{\powerS{0}C_{t_0}}{y^{\alpha_{t_0}}}\right)}
p_t(u)u\intd u \vphantom{\int_0^\infty}
\right)\\
&\stackrel{(c)}{=}\prod\nolimits_{t\in\{\L,\N\}}
\exp\left(
	-2\pi\lambdaST\int_0^\infty\expects{\XX{}}
{\ind{}\left(\fracS{\XX{}C_t}{u^{\alpha_t}}
\DCPbreakIII\DCspace
>\fracS{\XX{}_{0}C_{t_0}}{y^{\alpha_{t_0}}}\right)}
	p_t(u)u\intd u \vphantom{\int_0^\infty}
\right)
\end{align*}\ens
where $(a)$ is due to independence of LOS and NLOS tiers, $(b)$ is from PGFL of PPP and $(c)$ is due to the fact that $\powerS{}=\XX \IT$.
Now using the transformation $u=(\XX C_t)^{\frac{1}{\alpha_t}}z$, we get
\begs\begin{align}
\iftoggle{SC}{}{&}\prob{e(\y,\PPST)=1}\iftoggle{SC}{&}{}
\DCPbreak
=\prod_{t\in\{\L,\N\}}\exp\left(-2\pi\lambdaST\int_0^\infty
\expects{\XX}{\ind{}\left(\frac{1}{z^{\alpha_t}}>\frac{\XX{}_{0}C_{t_0}}{y^{\alpha_{t_0}}}\right)
\DCPbreakII \DCspace p_t((\XX C_t)^{\frac{1}{\alpha_t}}z)(\XX {}C_t)^{\frac{2}{\alpha_t}}z\intd z}
\right)\nonumber\\
&=\prod_{t\in\{\L,\N\}}\exp\left(-2\pi\lambdaST\int_0^{{\left(\fracS{\XX{}_{0}C_{t_0}}{y^{\alpha_{t_0}}}\right)}^{-\frac1{\alpha_t}}}
\imdfuni{t}{z}
z\intd z\right).\label{eq:app1e2}
\end{align}\ens
Using the value from \eqref{eq:app1e2}, \eqref{eq:app1e1} can be written as
\begs\begin{align*}
&\Pc^\mathrm{c}_\mathrm{S}(\SThres)=\sum_{t_0\in\{\L,\N\}}\mathbb{E}_{\XX_0}\left[\int_0^\infty\lambda p_{t_0}(y)\mathbb{P}\left[
\frac{\AntGain{\su}{1} \ChannelSS{} \XX_{0} \IT C_{t_0}y^{-\alpha_{t_0}}}{I_\pu+I_\su+\NoiseS} \DCPbreakII > \SThres | \left(\frac{\XX_{0}\IT C_{t_0}}{y^{\alpha_{t_0}}}> \frac{\XX_{j}\IT C_{t_j}}{y^{\alpha_{t_j}}} \forall j\in\Phi_2\right)\right]\right.\\
&
\exp\left(-2\pi\lambdaST\int_0^{\left(\fracS{\XX_0C_{t_0}}{y^{\alpha_{t_0}}}\right)^{-\frac1{\alpha_\L}}}
\imdfuni{\L}{z}
z\intd z\right)
\DCPbreak\left.
\exp\left(-2\pi\lambdaST\int_0^{\left(\fracS{\XX_0C_{t_0}}{y^{\alpha_{t_0}}}\right)^{-\frac1{\alpha_\N}}}
\imdfuni{\N}{z}
z\intd z\right)
2\pi y\intd y\right].
\end{align*}\ens
Now, substituting $y=u\XX_{0}^{1/\alpha_{t_0}}C_{t_0}^{1/\alpha_{t_0}} $, we get
\begs\begin{align*}
\Pc^\mathrm{c}_\mathrm{S}(\SThres)&=\sum_{t_0\in\{\L,\N\}}\mathbb{E}_{\XXi{0}}\left[\int_0^\infty\lambda p_{t_0}(u\XXi{0}^{1/\alpha_{t_0}}C_{t_0}^{1/\alpha_{t_0}})
\DCPbreakI
\mathbb{P}\left[
\frac{\AntGain{\su}{1}\IT \ChannelSS{} u^{-\alpha_{t_0}}}{I_\pu+I_\su+\NoiseS}>\SThres \bigconditioned
 \left(u^{\alpha_{t_0}}< \frac{y^{\alpha_{t_j}}}{\XXi{j}C_{t_j}}\forall j\in\Phi_2\right)\right]\right.\\
&
\exp\left(-2\pi\lambdaST\int_0^{u^{\fracS{\alpha_{t_0}}{\alpha_\L}}}\imdfuni{\L}{z}z\intd z\right)\iftoggle{SC}{}{\XXi{0}^{2/\alpha_{t_0}}{C_\L}^{2/\alpha_{t_0}}}
\DCPbreak
\left.
\exp\left(-2\pi\lambdaST\int_0^{u^{\fracS{\alpha_{t_0}}{\alpha_\N}}}\imdfuni{\N}{z}z\intd z\right)
\iftoggle{SC}{\XXi{0}^{2/\alpha_{t_0}}{C_\L}^{2/\alpha_{t_0}}}{}2\pi u\intd u\right]
\end{align*}\ens
which can be further simplified by moving the expectation inside as $\iftoggle{SC}{}{\Pc^\mathrm{c}_\mathrm{S}(\SThres)=}$
\begs\begin{align}
&\iftoggle{SC}{\Pc^\mathrm{c}_\mathrm{S}(\SThres)=}{}\sum_{t_0\in\{\L,\N\}}
\int_0^\infty2\pi\lambdaST\mathbb{E}_{\XXi{0}}\left[ p_{t_0}(u\XXi{0}^{1/\alpha_{t_0}}C_{t_0}^{1/\alpha_{t_0}})\XXi{0}^{2/\alpha_{t_0}}\right]
\DCPbreak
\mathbb{P}\left[
\frac{\AntGain{\su}{1} \IT \ChannelSS{} u^{-\alpha_{t_0}}}{I_\pu+I'_\su+\NoiseS}>\SThres 
\right]\iftoggle{SC}{\nonumber\\
&}{}
\exp\left(-2\pi\lambdaST\int_0^{u^{\fracS{\alpha_{t_0}}{\alpha_\L}}}
\imdfuni{\L}{z}z\intd z\right)
\DCPbreak \exp\left(-2\pi\lambdaST\int_0^{u^{\fracS{\alpha_{t_0}}{\alpha_\N}}}\imdfuni{\N}{z}z\intd z\right)
{C_\L}^{2/\alpha_{t_0}}u\intd u\label{eq:laststep1}
\end{align}\ens
where $I'_\su$ is interference from the conditioned secondary PPP. Using the MGF of $\ChannelSS{}$, the inner SINR probability term can be written as
\begs\begin{align}
&\mathbb{P}\left[
\frac{\AntGain{\su}{1} \ChannelSS{} \IT u^{-\alpha_{t_0}}}{I_\pu+I'_\su+\NoiseS}>\SThres
\right]\DCPbreak
=\exp\left(-\frac{\SThres u^{\alpha_{t_0}} }{\IT \AntGain{\su}{1} }\NoiseS\right)
\laplace{I_\pu}\left(\frac{\SThres u^{\alpha_{t_0}} }{ \IT\AntGain{\su}{1}  }\right)
\laplace{I'_\su}\left(\frac{\SThres u^{\alpha_{t_0}} }{ \IT\AntGain{\su}{1} }\right)\label{eq:laststep2}
\end{align}\ens
Using the definition of $\imdfuni{t}{z}$ and substituting \eqref{eq:laststep2} in \eqref{eq:laststep1}, we get the Lemma.

\section{Proof of Lemma \ref{lemma:SecPerSI}: \Secondary Interference at \UES}\label{app:SecPerSI}
The  interference from the conditional \secondary PPP is given as 
\begs\begin{align}
I'_\su=
	\sum_{\y_i\in \PPST} 
			\ChannelSS{i} \indside{\frac{y_i^{\alpha_{\LinkSS{i}}}}{\XXi{i}C_{\LinkSS{i}}}>u^{\alpha_{t_0}}}
			\IT\XXi{i}C_{\LinkSS{i}} y_i^{-\alpha_{\LinkSS{i}}}\AntGain{\su}{}(\theta_i).\nonumber
\end{align}\ens
$I'_\su$ can be split into interference from LOS and NLOS BSs in $\PPST$ as $I'_\su=I'_{\su\L}+ {I'_{\su\N}}$. Hence, the Laplace transform of $I'_\su$  can be expressed as product of Laplace transforms of $I'_{\su\L}$ and $I'_{\su\N}$. Now, the Laplace transform of  $I'_{\su\L}$ is given as 
\begs\begin{align*}
&\laplace{I'_{\su\L}}(s)=\expect{\exp
	\left(
		-s\sum_{\y_i\in {\PPST}_{\L}} 
			\ChannelSS{i}\indside{\frac{y_i^{\alpha_\L}}{\XXi{i}C_\L}>u^{\alpha_{t_0}}}
			\DCPbreakII\DCspace\DCspace\vphantom{\sum_{\y_i\in {\PPST}_{\L}} }
			\IT\XXi{i}C_\L y_i^{-\alpha_\L}\AntGain{\su}{}(\theta_i) 
	\right)}
	\\
&\stackrel{(a)}{=}\exp\left(
	-\lambdaST 2\pi \expects{\powerS{},\theta}{
		\int_{0}^\infty
			\left(1-
			\expU{
				-s \ChannelSS{} \IT\XXi{} C_\L y^{-\alpha_\L}\AntGain{\su}{}(\theta)}\right) \DCPbreakII\DCspace
				\indside{\frac{y^{\alpha_\L}}{\XXi{}C_\L}>u^{\alpha_{t_0}}}
			p_\L(y) y \intd y
	}
\right)
\end{align*}\ens
where $(a)$ is due to PGFL of the PPP. Now using the transformation $y=v(\XXi{}C_\L)^{1/\alpha_\L}$, we get
\begs\begin{align*}
\laplace{I'_{\su\L}}(s)
&=\exp\left(
	-\lambdaST 2\pi \expects{\powerS{},\theta}{
		\int_{0}^\infty
			\left(1-\expU{
				-s\ChannelSS{}\IT v^{-\alpha_\L}\AntGain{\su}{}(\theta) 		
			}\right)\DCPbreakII
			\indside{v^{\alpha_\L}>u^{\alpha_{t_0}}}
		p_\L(v(\XXi{}C_\L)^{\frac1{\alpha_\L}}) (\XXi{}C_\L)^{\frac2{\alpha_\L}} v \intd v
	}
\right).
\end{align*}\ens
Now moving  the expectation with respect to $\theta$ and $\XXi{}$ inside the integration, we get
\begs\begin{align}
\laplace{I'_{\su\L}}(s)&=\exp\left(
	-\lambdaST 2\pi 
		\int_{u^{\frac{\alpha_{t_0}}{\alpha_\L}}}^\infty
			\expects{\theta}{1-\expU{
				-s \ChannelSS{}\IT v^{-\alpha_\L}\AntGain{\su}{}(\theta) 		
			}}
			\DCPbreakI\DCspace
			\imdfuni{\L}{v} v \intd v
\right).\label{eq:laplaceSCellularStep4}
\end{align}\ens
Now, using definition of $a_k$'s, the inner term can be written as
\begs\begin{align}
&1-\expect{
	\exp\left(-s \ChannelSS{} \IT v^{-\alpha_\L}\AntGain{\su}{}(\theta) \right)
	}
=\expect{
	\frac {s \IT v^{-\alpha_\L}\AntGain{\su}{}(\theta) }
	 	{1+s\IT v^{-\alpha_\L}\AntGain{\su}{}(\theta) }
	}
	\DCPbreak
=\sum_{k=1}^2
	\frac {a_k }
		{1+\IT^{-1}s^{-1} \AntGain{\su}{k}^{-1}v^{\alpha_\L}}.\label{eq:mgfstepC}
\end{align}\ens
Using \eqref{eq:mgfstepC} in  \eqref{eq:laplaceSCellularStep4}, we get $\iftoggle{SC}{}{\laplace{I_{\su\L}}(s)=}$
\begs\begin{align}
\iftoggle{SC}{\laplace{I_{\su\L}}(s)=}{}
&\exp\left(
	-\lambdaST2\pi \int_{u^{\frac{\alpha_{t_0}}{\alpha_\L}}}^\infty
		\sum_{k=1}^2
			\frac {a_k }
				{1+\IT^{-1}s^{-1} \AntGain{\su}{k}^{-1}v^{\alpha_\L}}
			\imdfuni{\L}{v} v\intd v
\right).\nonumber
\end{align}\ens
Similarly $\laplace{I'_{\su\N}}(s)$ can be computed. Multiplying the values of $\laplace{I'_{\su\L}}(s)$  and $\laplace{I'_{\su\N}}(s)$ and using the definition of $F_\su(B,e)$, we get the Lemma.

\section{Proof of Lemma \ref{lemma:SecPerPI}: \Primary Interference at \UES}\label{app:SecPerPI}
The  \primary interference is given as $I_\pu=
	\sum_{\x_i\in \PPPT } 
			\ChannelPS{i}\powerP{} C_{\LinkPS{i}} x_i^{-\alpha_{\LinkPS{i}}}\AntGain{\pu}{}(\theta_i)$.
Similar to Appendix \ref{app:SecPerSI}, $\laplace{I_\pu}(s)=\laplace{I_{\pu\L}}(s)\laplace{I_{\pu\N}}(s)$. Using the PPP's PGFL, $\laplace{I_{\pu\L}}(s)$ can be computed as
\begs\begin{align*}
\laplace{I_{\pu\L}}(s)
&=\exp\left(
	-\lambdaPT 2\pi \expects{\theta}{
		\int_{0}^\infty
			\left(1-
			\expU{
				-s \ChannelPS{i}\powerP{} C_\L x^{-\alpha_\L}\AntGain{\pu}{}(\theta) 
			}
			\right)\DCPbreakII \DCspace \vphantom{\int_{0}^\infty}
			p_\L(x) x \intd x
	}
\right).
\end{align*}\ens
Now using the transformation $x=v(\powerP{}C_\L)^{1/\alpha_\L}$, we get
\begs\begin{align*}
\laplace{I_{\pu\L}}(s)
&=\exp\left(
	-\lambdaPT 2\pi \expects{\powerS{},\theta}{
		\int_{0}^\infty
			\left(1-\expU{
				-s\ChannelPS{i}v^{-\alpha_\L}\AntGain{\pu}{}(\theta)
			}\right)\DCPbreakII
		p_\L(v(\powerP{i}C_\L)^{\frac1{\alpha_\L}}) (\powerP{}C_\L)^{\frac2{\alpha_\L}} v \intd v
	}
\right).
\end{align*}\ens
Now, interchanging the order of  expectation and integration and using  $\imdfunP{t}{\cdot}$'s definition, we get
\begs\begin{align}
\laplace{I_{\pu\L}}(s)&=\exp\left(
	-\lambdaPT 2\pi 
		\int_{0}^\infty
			\expects{\theta}{1-\expU{
				-s \ChannelPS{} v^{-\alpha_\L}\AntGain{\pu}{}(\theta) 
			}}\DCPbreakI
			\imdfunP{\L}{v}  v \intd v
\right).\label{eq:laplacePCellularStep4}
\end{align}\ens
Now, using definition of  $b_k$'s, the inner term can be written as
\begs\begin{align}
&1-\expect{
	\exp\left(-s \ChannelPS{}v^{-\alpha_\L}\AntGain{\pu}{}(\theta) \right)
	}
=\expect{
	\frac {s v^{-\alpha_\L}\AntGain{\pu}{}(\theta) }
	 	{1+sv^{-\alpha_\L}\AntGain{\pu}{}(\theta) }
	}\DCPbreak
=\sum_{k=1}^2
	\frac {b_k }
		{1+s^{-1} \AntGain{\pu}{k}^{-1}v^{\alpha_\L}}\label{eq:tempstep1}.
\end{align}\ens
\iftoggle{SC}{\vspace{-.18in}
\begs\begin{align}
&\text{\begin{normalsize}Using \eqref{eq:tempstep1} in  \eqref{eq:laplacePCellularStep4}, we get\end{normalsize} } \laplace{I_{\pu\L}}(s)=
\exp\left(
	-\lambdaST2\pi \int_{0}^\infty
		\sum_{k=1}^2
			\frac {b_k }
				{1+s^{-1} \AntGain{\pu}{k}^{-1}v^{\alpha_\L}}
			\imdfunP{\L}{v}
			 v\intd v
\right).\nonumber\hspace{2in}
\end{align}\ens}
{Using \eqref{eq:tempstep1} in  \eqref{eq:laplacePCellularStep4}, we get $\iftoggle{SC}{}{\laplace{I_{\pu\L}}(s)=}$
\begs\begin{align}
\iftoggle{SC}{\laplace{I_{\pu\L}}(s)=}{}
&\exp\left(
	-\lambdaST2\pi \int_{0}^\infty
		\sum_{k=1}^2
			\frac {b_k }
				{1+s^{-1} \AntGain{\pu}{k}^{-1}v^{\alpha_\L}}
			\imdfunP{\L}{v}
			 v\intd v
\right).\nonumber
\end{align}\ens}
Similarly $\laplace{I_{\pu\N}}(s)$ can be computed. Using the values of $\laplace{I_{\pu\L}}(s)$  and $\laplace{I_{\pu\N}}(s)$ and the definition of $F_\pu(B)$, we get the Lemma.

\section{Proof of Lemma \ref{lemma:PriSI}: \Secondary Interference at \UEP}\label{app:PriSI}
Let us first consider ${I_{\foreign}}$  which is given as
\begs\begin{align}
{I_{\foreign}}(s)=&
		\sum_{t\in\{\L,\N\}}\sum_{\y_i\in {\PPST}_t} 
			\ChannelSP{i} \indside{C_{\LinkHome{i}}R_i^{-\alpha_{\LinkHome{i}}}>C_{t} y_i^{-\alpha_{t}}}
			\DCPbreak\iftoggle{SC}{}{\hspace{1.2in}\times}
			\IT\XXi{i}C_{t} y_i^{-\alpha_{t}}\AntGain{\su}{}(\theta_i) 
\end{align}\ens
where the indicator term denotes that only those  \secondary BSs are considered whose  receiver power at the their home \primary user is greater than their received power at \UEP~which means that \UEP~is not the home \primary user for these BSs. Now its Laplace transform  is equal to
\begs\begin{align}
&\laplace{I_{\foreign}}(s)\stackrel{(a)}{=}\prod_{t{}\in\{\L,\N\}}\expect{\exp
	\left(
		-s\sum_{\y_i\in  {\PPST}_t} 
			\ChannelSP{i} \indside{1/\XXi{}>C_{t{}} y_i^{-\alpha_{t{}}}}
			\DCPbreakII
			\iftoggle{SC}{}{\hspace{1.2in}\times}
\IT\XXi{i}C_{t{}} y_i^{-\alpha_{t{}}}\AntGain{\su}{}(\theta_i) 
	\right)}\nonumber\\
&\stackrel{(b)}{=}\prod_{t{}}\exp\left(
	-\lambdaST 2\pi \expects{\XXi{},\theta}{
		\int_{0}^\infty\left(
			1-
			\expU{
				-s \ChannelSP{} \IT \XXi{} C_{t{}} y^{-\alpha_{t{}}}\AntGain{\su}{}(\theta) 
			}\right)\DCPbreakII\iftoggle{SC}{}{\hspace{1.2in}\times}
			\indside{\XXi{}<C_{t{}}^{-1} y_i^{\alpha_{t{}}}}
			p_t(y) y \intd y
	}
\right)\nonumber
\end{align}\ens
where $(a)$ is from independence of LOS and NLOS tiers and $(b)$ is due to the PGFL of PPP. Now, using the transformation $y=v(\XXi{} C_{t{}})^{1/\alpha_{t{}}}$, we get $\laplace{I_{\foreign}}(s)=$
\begs\begin{align}
&\prod_{t{}}\exp\left(
	-\lambdaST 2\pi \expects{\XXi{},\theta}{
		\int_{0}^\infty\left(
			1-
			\exp\left(
				-s \ChannelSP{} \IT  v^{-\alpha_{t{}}}\AntGain{\su}{}(\theta) \right)
				\DCPbreakIII\iftoggle{SC}{}{\hspace{.2in}\times}
				\indside{v>1}
			\right)
			p_{t}(v(\XXi{} C_{t{}})^{\frac1{\alpha_{t{}}}}) v(\XXi{} C_{t{}})^{\frac{2}{\alpha_{t{}}}}  \intd v
	}
\right)\nonumber\\
&\stackrel{(a)}{=}\prod_{t{}}\exp\left(
	-\lambdaST 2\pi \expects{\theta}{
		\int_{1}^\infty\left(
			1-
			\exp\left(
				-s \ChannelSP{} \IT  v^{-\alpha_{t{}}}\AntGain{\su}{}(\theta) 	
			\right)\right)\DCPbreakII\iftoggle{SC}{}{\hspace{.2in}\times}
			\imdfuni{t{}}{v} v\intd v
	}
\right)\nonumber
\end{align}\ens
where $(a)$ is due to interchanging the integration and the expectation with respect to $\XXi{}$  and applying $K_t$'s definition. Now, using the MGF of   $\ChannelSP{} $ and the distribution of $\AntGain{\su}{}(\theta)$, we get \iftoggle{SC}{}{$\laplace{I_{\foreign}}(s)=$}
\begs\begin{align}
\iftoggle{SC}{\laplace{I_{\foreign}}(s)=}{}
&\prod_{t{}}\exp\left(
	-\lambdaST2\pi \int_{1}^\infty
		\sum_{k=1}^2
			\frac {a_k }
				{1+s^{-1}\IT^{-1} \AntGain{\su}{k}^{-1}v^{\alpha_{t{}}}}
			\imdfuni{t{}}{v} v\intd v
\right)\nonumber.
\end{align}\ens
Now substituting $u={(s\IT\AntGain{\su}{k})}^{-1/\alpha_{t{}}} v$, we get
\begs\begin{align}
\laplace{I_{\foreign}}(s)
&=\prod_{t{}}\exp\left(
	-\lambdaST2\pi {(s\IT\AntGain{\su}{k})}^{2/\alpha_{t{}}} 
	\DCPbreakI
	\int_{{(s\IT\AntGain{\su}{k})}^{-1/\alpha_{t{}}}}^\infty
		\sum_{k=1}^2
			\frac {a_k }
				{1+u^{\alpha_{\LinkSP{}}}}
			\imdfuni{t{}}{{(s\IT\AntGain{\su}{k})}^{1/\alpha_{t{}}}u} u\intd u
\right)\nonumber.
\end{align}\ens
Using the definition of  $E_{\foreign }(B,\IT)$, we get
\begs\begin{align}
\laplace{I_{\foreign}}(s)&=
		\exp\left(-\lambdaST  \sum_{k=1}^2a_k
E_{\foreign}\left(s  \AntGain{\su}{k} ,\IT\right)
\right).\label{eq:foreignformula}
\end{align}\ens
Now, let us consider ${I_{\native}}$ which is given as \iftoggle{SC}{}{$\laplace{I_{\native}}(s)=$}
\begs\begin{align}
\iftoggle{SC}{{I_{\native}}(s)=}{}&
		\sum_{t{}\in\{\L,\N\}}\sum_{i\in {\PPST}_\L} 
			\ChannelSP{i} \left(1-\indside{C_{\LinkHome{i}}R_i^{-\alpha_{\LinkHome{i}}}>C_\L y_i^{-\alpha_L}}\right)
			\IT\AntGain{\su}{}(\theta_i) \nonumber
\end{align}\ens
where the indicator term are exact opposite of the previous case and denotes that only those secondary BSs are considered whose receiver power at the their home \primary user is not greater than their received power at \UEP. 
Note that the interference from each of these secondary BSs is equal to $\IT$. Hence, its Laplace transform  is given as $\iftoggle{SC}{}{\laplace{I_{\native}}(s)=}{}$
\begs \begin{align}
&\iftoggle{SC}{\laplace{I_{\native}}(s)=}{}\prod_{t\in\{\L,\N\}}\expect{\exp
	\left(
		-s\sum_{i\in {\PPST}_\L} 
			\ChannelSP{i}
			\indside{
			\iftoggle{SC}{C_{T_i}
			R_i^{-\alpha_{T_i}}}{\frac{C_{T_i}}{R_i^{\alpha_{T_i}}}}
			<
			\iftoggle{SC}{C_t y_i^{-\alpha_t}}{\frac{C_t}{ y_i^{\alpha_t}}}}
			\IT\AntGain{\su}{}(\theta_i) 
	\right)}\nonumber\\
&\stackrel{(a)}{=}\prod_{t\in\{\L,\N\}}\exp\left(
	-\lambdaST 2\pi \expects{\XXi{},\theta}{
		\int_{0}^\infty
			\left(1-
			\exp\left(
				-s \ChannelSP{}\IT\AntGain{\su}{}(\theta) 
				\DCPbreakIV\iftoggle{SC}{}{\hspace{1in}\times}				\indside{\frac{y^{\alpha_t}}{\XXi{}C_t}<1}
			\right)\right)
			p_t(y) y \intd y
	}
\right)\nonumber
\end{align}\ens
where $(a)$ is due to PGFL of a PPP. Substituting $y=(\XXi{}C_t)^{1/\alpha_t}v$, we get $\laplace{I_{\native}}(s)=$
\begs\begin{align}
&\prod_{t\in\{\L,\N\}}\exp\left(
	-\lambdaST 2\pi \expects{\XXi{},\theta}{
		\int_{0}^\infty
			\left(1-
			\expU{
				-s \ChannelSP{}\IT\AntGain{\su}{}(\theta) 
			}\right)
			\DCPbreakII\DCspace\DCspace
			\indside{v^{\alpha_t}<1}
			p_t((\XXi{}C_t)^{\frac1{\alpha_t}}v) (\XXi{}C_t)^{\frac2{\alpha_t}}v \intd v
	}
\right).\nonumber
\end{align}\ens
Now using the MGF of exponential $\ChannelSP{}$ and the PMF of $\AntGain{\su}{}(\theta) $, we get 
\begs\begin{align}
\laplace{I_{\native}}(s)
&=\exp\left(
	-\lambdaST 2\pi \left[\sum_{k=1}^2{\frac{a_k}{
			1+
				(s \IT\AntGain{\su}{k} )^{-1} }}\right]
				\DCPbreakI
				\left(
		\int_0^{1}
			\imdfuni{\L}{v}v \intd v+
			\int_0^{1}
			\imdfuni{\N}{v}v \intd v
			\right)
\right).\label{eq:nativeformula}
\end{align}\ens
Using  \eqref{eq:foreignformula} and \eqref{eq:nativeformula}, we get the Lemma.

\linespread{1.2}


\begin{thebibliography}{10}
\providecommand{\url}[1]{#1}
\csname url@samestyle\endcsname
\providecommand{\newblock}{\relax}
\providecommand{\bibinfo}[2]{#2}
\providecommand{\BIBentrySTDinterwordspacing}{\spaceskip=0pt\relax}
\providecommand{\BIBentryALTinterwordstretchfactor}{4}
\providecommand{\BIBentryALTinterwordspacing}{\spaceskip=\fontdimen2\font plus
\BIBentryALTinterwordstretchfactor\fontdimen3\font minus
  \fontdimen4\font\relax}
\providecommand{\BIBforeignlanguage}[2]{{%
\expandafter\ifx\csname l@#1\endcsname\relax
\typeout{** WARNING: IEEEtran.bst: No hyphenation pattern has been}%
\typeout{** loaded for the language `#1'. Using the pattern for}%
\typeout{** the default language instead.}%
\else
\language=\csname l@#1\endcsname
\fi
#2}}
\providecommand{\BIBdecl}{\relax}
\BIBdecl

\bibitem{GuptaAlkhateeb2016}
A.~K. Gupta, A.~Alkhateeb, J.~G. Andrews, and R.~W. Heath~Jr, ``Restricted
  secondary licensing in millimeter wave cellular system: How much gain can be
  obtained?'' submitted to IEEE GLOBECOM.

\bibitem{PiKhan2011}
Z.~Pi and F.~Khan, ``An introduction to millimeter-wave mobile broadband
  systems,'' \emph{IEEE Commun. Mag.}, vol.~49, no.~6, pp. 101--107, June 2011.

\bibitem{Andrews5G}
J.~Andrews, S.~Buzzi, W.~Choi, S.~Hanly, A.~Lozano, A.~Soong, and J.~Zhang,
  ``What will 5{G} be?'' \emph{IEEE J. Sel. Areas Commun.}, vol.~32, no.~6, pp.
  1065--1082, June 2014.

\bibitem{Boccardi2014}
F.~Boccardi, R.~Heath, A.~Lozano, T.~Marzetta, and P.~Popovski, ``Five
  disruptive technology directions for {5G},'' \emph{IEEE Commun. Mag.},
  vol.~52, no.~2, pp. 74--80, Feb. 2014.

\bibitem{Rangan2014}
S.~Rangan, T.~Rappaport, and E.~Erkip, ``Millimeter-wave cellular wireless
  networks: Potentials and challenges,'' \emph{Proc. IEEE}, vol. 102, no.~3,
  pp. 366--385, March 2014.

\bibitem{SinghBackHaul2015}
S.~Singh, M.~Kulkarni, A.~Ghosh, and J.~Andrews, ``Tractable model for rate in
  self-backhauled millimeter wave cellular networks,'' \emph{IEEE J. Sel. Areas
  Commun.}, vol.~PP, no.~99, pp. 1--1, 2015.

\bibitem{Bai2014}
T.~Bai and R.~W. Heath~Jr., ``Coverage and rate analysis for millimeter wave
  cellular networks,'' \emph{IEEE Trans. Wireless Commun.}, vol.~14, no.~2, pp.
  1100--1114, Feb. 2015.

\bibitem{GuptaAndHeath2016}
A.~K. Gupta, J.~G. Andrews, and R.~W. Heath~Jr, ``On the feasibility of sharing
  spectrum licenses in mm{W}ave cellular systems,'' \emph{submitted to IEEE
  Trans. Commun., arXiv preprint arXiv:1512.01290}, 2016.

\bibitem{FCC2002}
FCC, ``Spectrum policy task force,'' \emph{ET Docket 02-135}, Nov. 2002.

\bibitem{Haykin2005}
S.~Haykin, ``Cognitive radio: brain-empowered wireless communications,''
  \emph{IEEE J. Sel. Areas Commun.}, vol.~23, no.~2, pp. 201--220, Feb 2005.

\bibitem{Kang2009}
X.~Kang, Y.-C. Liang, H.~Garg, and L.~Zhang, ``Sensing-based spectrum sharing
  in cognitive radio networks,'' \emph{IEEE Trans. Veh. Technol.}, vol.~58,
  no.~8, pp. 4649--4654, Oct. 2009.

\bibitem{Stevenson2009}
C.~Stevenson, G.~Chouinard, Z.~Lei, W.~Hu, S.~Shellhammer, and W.~Caldwell,
  ``{IEEE} 802.22: The first cognitive radio wireless regional area network
  standard,'' \emph{IEEE Commun. Mag.}, vol.~47, no.~1, pp. 130--138, Jan.
  2009.

\bibitem{Akyildiz06}
I.~F. Akyildiz, W.-Y. Lee, M.~C. Vuran, and S.~Mohanty, ``Ne{X}t
  generation/dynamic spectrum access/cognitive radio wireless networks: A
  survey,'' \emph{Computer Networks}, pp. 2127--2159, 2006.

\bibitem{Stotas2011}
S.~Stotas and A.~Nallanathan, ``Enhancing the capacity of spectrum sharing
  cognitive radio networks,'' \emph{IEEE Trans. Veh. Technol.}, vol.~60, no.~8,
  pp. 3768--3779, Oct 2011.

\bibitem{Lima2012}
C.~Lima, M.~Bennis, and M.~Latva-aho, ``Coordination mechanisms for
  self-organizing femtocells in two-tier coexistence scenarios,'' \emph{{IEEE}
  Trans. Wireless Commun.}, vol.~11, no.~6, pp. 2212--2223, June 2012.

\bibitem{ElSawy2014}
H.~ElSawy and E.~Hossain, ``Two-tier {HetNets} with cognitive femtocells:
  Downlink performance modeling and analysis in a multichannel environment,''
  \emph{IEEE Trans. Mobile Computing}, vol.~13, no.~3, pp. 649--663, March
  2014.

\bibitem{Khosh2013}
M.~Khoshkholgh, K.~Navaie, and H.~Yanikomeroglu, ``Outage performance of the
  primary service in spectrum sharing networks,'' \emph{IEEE Trans. Mobile
  Computing}, vol.~12, no.~10, pp. 1955--1971, Oct. 2013.

\bibitem{Nguyen2012}
T.~V. Nguyen and F.~Baccelli, ``A stochastic geometry model for cognitive radio
  networks,'' \emph{The Computer Journal}, vol.~55, no.~5, pp. 534--552, 2012.

\bibitem{Bae2008}
J.~Bae, E.~Beigman, R.~Berry, M.~L. Honig, H.~Shen, R.~Vohra, and H.~Zhou,
  ``Spectrum markets for wireless services,'' in \emph{Proc. IEEE DySPAN}, Oct
  2008, pp. 1--10.

\bibitem{Guo2015}
A.~Guo and M.~Haenggi, ``Asymptotic deployment gain: A simple approach to
  characterize the {SINR} distribution in general cellular networks,''
  \emph{IEEE Trans. Commun.}, vol.~63, pp. 962--976, Mar. 2015.

\bibitem{GantiArxiv}
R.~K. Ganti and M.~Haenggi, ``Asymptotics and approximation of the {SIR}
  distribution in general cellular networks,'' \emph{arXiv preprint
  arXiv:1505.02310v1}.

\bibitem{DaSilva2015}
J.~Kibilda, P.~D. Francesco, F.~Malandrino, and L.~A. DaSilva, ``Infrastructure
  and spectrum sharing tradeoffs in mobile networks,'' in \emph{Proc. IEEE
  DySPAN}, Stockholm, Sweden, Sept. 2015, pp. 348--357.

\bibitem{Akoum2012}
S.~Akoum, O.~El~Ayach, and R.~W. Heath, ``Coverage and capacity in mmwave
  cellular systems,'' in \emph{Proc. ASILOMAR}, Pacific Grove, CA, 2012, pp.
  688--692.

\bibitem{Hunter2008}
A.~M. Hunter, J.~G. Andrews, and S.~Weber, ``Transmission capacity of ad hoc
  networks with spatial diversity,'' \emph{IEEE Trans. on Wireless Commun.},
  vol.~7, no.~12, pp. 5058--5071, December 2008.

\bibitem{AndGupDhi2016}
J.~G. Andrews, A.~K. Gupta, and H.~S. Dhillon, ``A primer on cellular network
  analysis using stochastic geometry,'' \emph{submitted to IEEE Commun. Surveys
  Tuts., arXiv preprint arXiv:1604.03183}, 2016.

\bibitem{VanTrees2002}
H.~L. Van~Trees, ``Optimum array processing (detection, estimation, and
  modulation theory, part iv),'' \emph{Wiley-Interscience, Mar}, no.~50, p.
  100, 2002.

\bibitem{SinDhiJ2013}
S.~Singh, H.~S. Dhillon, and J.~G. Andrews, ``Offloading in heterogeneous
  networks: modeling, analysis and design insights,'' \emph{IEEE Trans.
  Wireless Commun.}, vol.~12, no.~5, pp. 2484 -- 2497, May 2013.

\end{thebibliography}

\end{document}